\documentclass{article}
\usepackage{amssymb}
\usepackage{amsmath}
\usepackage{amsthm}
\usepackage{amsxtra}
\usepackage{bm}
\usepackage{graphicx}
\usepackage{graphics}
\usepackage{subfigure}
\usepackage[svgnames]{xcolor} 
\usepackage{color}
\usepackage[pdftex,pagebackref=false]{hyperref}
\hypersetup{colorlinks,citecolor=darkblue,linkcolor=darkblue,urlcolor=darkblue}
\definecolor{darkblue}{rgb}{0,0.08,0.45}
\usepackage{pgf,tikz}
\usetikzlibrary{arrows}
\usepackage[T1]{fontenc}%needed in order to use Polish Characters !!!
\usepackage{ae,aecompl}

\newcommand{\bigfig}[1]{} %#1
 
\usepackage{fullpage}
\usepackage[hmargin=1in,vmargin=1in]{geometry}

\newcommand{\floor}[1]{\lfloor #1 \rfloor}
\newcommand{\abs}[1]{\left| #1 \right|}
\newcommand{\comment}[1]{}
\newcommand{\set}[1]{\{#1\}}

\newtheorem{lemma}{Lemma}
\newtheorem{invariant}{Invariant}
\newtheorem{corollary}{Corollary}
\newtheorem{definition}{Definition}
\newtheorem{theorem}{Theorem}
\newtheorem*{theorem*}{Theorem}
\newtheorem{fact}{Fact}

\def\dualfacelevel{\ell^F}
\def\dualnodelevel{\ell^V}
\def\dualedgelevel{\ell^E}
\def\dualcomponentlevel{\ell^{\mathcal K}}

\newcommand{\parent}{\mbox{parent}}
\newcommand{\rev}{\text{rev}}

\newcommand{\barG}{\bar G}
\usepackage{xspace}

\usepackage[ruled,vlined,linesnumbered]{algorithm2e}

\begin{document}
\begin{titlepage}
\title{Structured Recursive Separator Decompositions\\ for Planar Graphs in Linear Time\thanks{Announced at the 45th ACM Symposium on Theory of Computing (STOC)~\cite{KMS13}.}}
\author{
Philip N.~Klein\thanks{Brown University, Providence, RI. Supported in part by National Science Foundation Grant CCF-0964037.}
\and Shay Mozes\thanks{MIT, Cambridge, MA. 
Supported in part by 
National Science Foundation Grants CCF-0964037 and CCF-1111109. 
Work partially conducted while at Brown University.}
\and Christian Sommer\thanks{Supported in part by 
the Swiss National Science Foundation. 
Work done while at MIT CSAIL, Cambridge, MA. }}
\date{}
\maketitle

\begin{abstract} 
Given a planar graph $G$ on $n$ vertices and an integer $r<n$, 
an {\em $r$--division} of $G$ with {\em few holes} is 
a decomposition of $G$ into $O(n/r)$ regions of size at most $r$
such that 
each region contains at most a constant number of faces that are not faces of $G$ 
(also called {\em holes}), 
and such that, 
for each region, 
the total number of vertices on these faces is $O(\sqrt r)$.

We provide a linear-time algorithm for 
computing $r$--divisions with few holes. 
In fact, 
our algorithm computes a structure, 
called {\em decomposition tree}, 
which represents a recursive decomposition of~$G$ 
that includes $r$--divisions for essentially all values of $r$.
In particular, 
given an increasing sequence ${\bm r} = (r_1,r_2,\dots)$,
our algorithm can produce 
a {\em recursive} ${\bm r}$--division with few holes 
in linear time.

$r$--divisions with few holes have been used 
in efficient algorithms to compute 
shortest paths, 
minimum cuts, and 
maximum flows. 
Our linear-time algorithm improves upon 
the decomposition algorithm 
used in the state-of-the-art algorithm for minimum $st$--cut 
(Italiano, Nussbaum, Sankowski, and Wulff-Nilsen, STOC~2011), 
removing one of the bottlenecks in the overall running time of their algorithm 
(analogously for minimum cut in planar and bounded-genus graphs). 
\end{abstract}
\thispagestyle{empty}

\end{titlepage}

\section{Introduction}

Separators decompose a graph in a {\em balanced} way into 
two subgraphs with a limited number of vertices in common. 
Separators are often used in 
efficient algorithms using a divide-and-conquer strategy~\cite{journals/siamcomp/LiptonT80,ChungSurvey,Shmoys,RosenbergHeath}. 
Graphs with small recursive separators include 
planar graphs \cite{Ungar51,journals/siamjam/LiptonT79,Djidjev1982,journals/jcss/Miller86,conf/sigal/GazitM90,conf/compgeom/SpielmanT96,DjidjevVenkatesan97}, 
bounded-genus graphs \cite{DjidjevSerdica1985LinearTime,journals/jal/GilbertHT84,journals/siamcomp/Kelner06}, 
minor-free graphs \cite{journals/jct/Andreae86,AlonST90,conf/soda/PlotkinRS94,journals/talg/ReedW09,journals/jacm/BiswalLR10,conf/focs/KawarabayashiR10,CWNMinorFree}, 
and 
graphs with bounded tree-width~\cite{halin,RS2,Arn}. 
Furthermore, 
for graphs of all these classes, 
separators can be found efficiently, 
often in linear time. 
For planar graphs, 
experimental results demonstrate that separator algorithms are practical~\cite{journals/jea/AleksandrovDGM06,journals/jea/HolzerSWPZ09,FMPS13}. 

Perhaps the most influential result of this kind is 
the linear-time algorithm of Lipton and Tarjan~\cite{journals/siamjam/LiptonT79} 
for finding a separator of size $O(\sqrt{n})$ in a planar graph with $n$ vertices.
Consider the result of using this algorithm recursively 
until each separated subgraph has size at most 
some specified limit~$r$. 
It is easy to show 
that $O(n/r)$ subgraphs result, 
and 
that the average number of boundary vertices per subgraph is $O(\sqrt{r})$. 
Frederickson~\cite{journals/siamcomp/Frederickson87} showed that, 
with additional care, 
one can ensure that each subgraph has $O(\sqrt{r})$ boundary vertices; 
he calls such a decomposition an {\em $r$--division}, 
and 
he refers to the subgraphs as {\em regions}. 
The running time of Frederickson's algorithm is $O(n\log r+(n/\sqrt{r})\log n)$. 

Decompositions of such kind have been used in 
many efficient algorithms for planar graphs,
such as those for computing 
shortest paths~\cite{journals/siamcomp/Frederickson87,journals/jcss/HenzingerKRS97},
maximum flow~\cite{JV82}, 
polygon triangulation~\cite{journals/jcss/Goodrich95}, 
and 
I/O--efficient algorithms~\cite{journals/siamcomp/MaheshwariZ08}.

For some of these applications, 
Frederickson's algorithm for $r$--divisions is too slow.
Goodrich~\cite{journals/jcss/Goodrich95} gave a linear-time algorithm 
that achieves $O(\sqrt{r})$ boundary size for planar graphs
(see~\cite{journals/jal/ArgeBT04} for the I/O model).
Even this was not enough for 
the linear-time shortest-path algorithm of Henzinger, Klein, Rao, and Subramanian~\cite{journals/jcss/HenzingerKRS97}, 
which requires recursive applications of an algorithm for $r$--divisions 
(with roughly $\log^* n$ levels of recursion). 
They addressed this by showing 
that a linear-time $O(\sqrt{n})$ separator algorithm can be used 
to obtain a sublinear-time separator algorithm with a worse 
(but still sublinear) 
boundary-size guarantee. 

However, 
for some of the more involved algorithms working with planar embedded graphs, 
it is essential that the regions of the division are topologically ``nice'' in 
that the boundary of each region consists of a constant number of faces, 
also called {\em holes}. 
Such a division can be found 
by using small {\em cycle} separators (Miller~\cite{journals/jcss/Miller86})
instead of just small separators, and 
by incorporating iterations in which 
the graph is separated according to the number of holes. 

Such $r$--divisions with a constant number of holes 
were first used in algorithms of Klein and Subramanian~\cite{journals/algorithmica/KleinS98,SubramanianThesis},
and subsequently in many other algorithms~\cite{journals/jcss/FakcharoenpholR06,CabelloR10,journals/algorithmica/Cabello12,esa/MozesW10,conf/stoc/ItalianoNSW11,conf/esa/LackiS11,KKS,conf/soda/EricksonFN12,MS12}.\footnote{Cabello~\cite{journals/algorithmica/Cabello12} requires only that the {\em average} number of holes per region be small.}

Prior to this work,
computing an $r$--division with a constant number of holes per region 
required time $O(n\log r+(n/\sqrt{r})\log n)$~\cite{conf/stoc/ItalianoNSW11}, 
Such $r$--divisions were one of the time bottlenecks in the state-of-the-art algorithms for 
minimum $st$--cut and maximum $st$--flow~\cite{conf/stoc/ItalianoNSW11} and 
minimum cut in undirected planar graphs~\cite{conf/esa/LackiS11} and 
bounded-genus graphs~\cite{conf/soda/EricksonFN12}. 
Whether such an $r$--division can be computed in linear time 
was an open problem until the current work.
For example, 
Cabello~\cite{journals/algorithmica/Cabello12} remarks that 
{\em ``it is unclear if the algorithm of Goodrich~\cite{journals/jcss/Goodrich95} can also be modified to use the cycle-separator, 
and thus obtain a linear-time construction of $r$--decompositions with a few holes.''}

\subsection*{Contributions}
We provide a linear-time algorithm for computing $r$--divisions
with few holes for any triangulated biconnected plane graph $G$ and
any~$r$ (Theorem~\ref{thm:r-division}).

In fact, the algorithm produces a {\em decomposition tree} of~$G$ (Theorem~\ref{thm:decomposition}),
which is a tree 
that naturally represents a recursive decomposition of $G$ by cycle separators, 
and 
from which one can read off, 
in linear time, 
a recursive ${\bm r}$--division with few holes, 
for any exponentially increasing sequence ${\bm r} = r_1,r_2, \dots$ (Theorem~\ref{thm:recursive-r-division}).

Our linear-time algorithm improves upon 
the $O(n\log r+(n/\sqrt{r})\log n)$--time algorithm of Italiano, Nussbaum, Sankowski, and Wulff-Nilsen~\cite{conf/stoc/ItalianoNSW11}, 
removing one of the time bottlenecks of the state-of-the-art algorithms for 
minimum $st$--cut and maximum $st$--flow~\cite{conf/stoc/ItalianoNSW11}, 
as well as
minimum cut~\cite{conf/esa/LackiS11,conf/soda/EricksonFN12} 
(all in undirected graphs).

Concurrent with and independent of our result, 
Arge, van Walderveen, and Zeh~\cite{AWZ13} also gave an algorithm 
that produces an $r$--division with few holes in linear time. 
The focus of their paper is I/O-efficient algorithms. 
Their approach differs from ours in that 
their algorithm requires the value of $r$ as input and 
uses it in a way that makes it unsuitable for 
producing recursive ${\bm r}$--divisions in linear time. 
Such a linear time algorithm seems to be crucial for 
obtaining linear-time algorithms using $r$--divisions (see~\cite{journals/jcss/HenzingerKRS97}).

\subsection*{Techniques} The overall approach of our algorithm builds on 
that of Goodrich~\cite{journals/jcss/Goodrich95}, 
which is based on Lipton-Tarjan vertex separators~\cite{journals/siamjam/LiptonT79}. 
However, 
our approach must handle the 
additional complexities 
of finding cycle separators (see Miller~\cite{journals/jcss/Miller86}), 
which involves 
maintaining spanning trees in both the input graph and its planar dual, 
and of bounding the number of holes.
We manage both primal and dual trees simultaneously by using
a primal spanning tree
that follows a dual breadth-first-search tree (Section~\ref{sec:cyclesep:mirrortree}). 
The levels of the dual BFS tree define connected components, 
which we maintain in a {\em Component Tree} (Section~\ref{sec:components}). 
This component tree internally captures 
the structural connectivity of dual BFS components. 

\subsection*{Outline}
Precise definitions of $r$--divisions, holes, recursive divisions, and decomposition trees, 
as well as formal statements of our results are 
in Section~\ref{sec:prel:rdiv} and Section~\ref{sec:prel:recursivediv}. 
We give a high-level description of the $r$--division algorithm 
and prove its correctness in Section~\ref{sec:rdiv}.

A key ingredient of our algorithm is 
a cycle separator algorithm (Section~\ref{sec:cyclesep}) 
that can be implemented efficiently by using auxiliary data structures 
(based on {\em dynamic trees} and {\em Euler-tour trees}, 
see Section~\ref{sec:prelim:dtree}).
In a straightforward computation of an $r$--division, 
the cycle separator at each recursive level can be found in linear time 
and the graph can be cut in linear time into two regions along the cycle. 
In order to achieve {\em overall} linear time, 
we aim at {\em sublinear} time bounds for 
computing the cycle separator 
and cutting the graph along it.
In our algorithm, 
we initialize the data structures once, 
and then, 
in subsequent steps of the recursion, 
we reuse what has been computed already (Section~\ref{sec:maint}), 
and we call an efficient dynamic-tree implementation of our cycle separator algorithm 
(Section~\ref{sec:cyclesep-eff}).

\section{Preliminaries}\label{sec:prel}

\subsection{Graph Notation}
\label{sec:prelim:planar}
Let $G=(V,E)$ be a simple graph. 
For a subset $V'$ of $V$, 
we denote by $\delta_G(V')$ the set of edges $uv$ of $G$ 
such that $v\in V'$ and $u\not\in V'$. 
We refer to $\delta_G(V')$ as a {\em cut} in $G$.
If $V'$ and $V-V'$ both induce connected subgraphs of $G$, 
we say it is a {\em simple cut}
(also known as a {\em bond}). 
In this case, $V'$ and $V-V'$ are the
vertex-sets of the two connected components of $\delta_G(V')$, which
shows that the edges of $\delta_G(V')$ uniquely determine the
bipartition $\set{V', V-V'}$.

A graph is called {\em biconnected} iff 
any pair of vertices is connected by at least two vertex-disoint paths.

For a spanning tree $T$ of $G$ and an edge $e$ of $G$ not in~$T$, 
the {\em fundamental cycle} of $e$ with respect to $T$ in $G$ is 
the simple cycle consisting of $e$ and 
the unique simple path in $T$ between the endpoints of $e$.

\subsection{Embeddings and Planar Graphs}

We review the definitions of combinatorial embeddings. 
For elaboration, see also \url{http://planarity.org}.

Let $E$ be a finite set, the edge-set. 
We define the corresponding set of {\em darts} to be $E \times \set{\pm 1}$. 
For $e\in E$, 
the darts of $e$ are $(e, +1)$ and $(e,-1)$. 
We think of the darts of $e$ as oriented versions of $e$ 
(one for each orientation).
We define the involution $\rev(\cdot)$ by $\rev((e,i)) = (e,-i)$. 
That is, 
$\rev(d)$ is the dart with the same edge but the opposite orientation.

A graph on $E$ is defined to be a pair $(V, E)$ 
where $V$ is a partition of the darts. 
Thus each element of $V$ is a nonempty subset of darts. 
We refer to the elements of $V$ as {\em vertices}.
The endpoints of an edge $e$ are the subsets $v, v'\in V$ containing the darts of $e$. 
The {\em head} of a dart $d$ is the subset $v\in V$ containing $d$, 
and its {\em tail} is the head of $\rev(d)$.

An {\em embedding} of $(V, E)$ is 
a permutation $\pi$ of the darts
such that $V$ is the set of orbits of $\pi$. 
For each orbit $v$, 
the restriction of $\pi$ to that orbit is a permutation cycle. 
The permutation cycle for $v$ specifies 
how the darts with head $v$ are arranged around $v$ in the embedding 
(in, say, counterclockwise order).
We refer to the pair $(\pi, E)$ as an embedded graph.

Let $\pi^*$ denote the permutation $\rev \circ \pi$, 
where $\circ$ denotes functional composition. 
Then $(\pi^*, E)$ is another embedded graph, 
the {\em dual} of $(\pi, E)$. 
(In this context, 
we refer to $(\pi, E)$ as the {\em primal.}) 

The {\em faces} of $(\pi, E)$ are defined to be the vertices of $(\pi^*, E)$. 
Since $\rev \circ (\rev\circ \pi) = \pi$, 
the dual of the dual of $(\pi, E)$ is $(\pi, E)$.
Therefore the faces of $(\pi^*, E)$ are the vertices of $(\pi, E)$.

We define an embedded graph $(\pi, E)$ to be {\em planar} 
if $n-m+\phi=2\kappa$, 
where $n$~is the number of vertices, 
$m$~is the number of edges, 
$\phi$~is the number of faces, 
and 
$\kappa$~is the number of connected components. 
Since taking the dual swaps vertices and faces 
and preserves the number of connected components, 
the dual of a planar embedded graph is also planar.
 
Note that, 
according to our notation, 
we can use $e$ to refer to an edge in the primal or the dual. 

\subsection{Some Properties of Planar Graphs}
\label{sec:planar}
\begin{fact}[Sparsity]\label{fact:sparse} 
Let $G$ be a simple planar graph. 
$\abs{E(G)} \leq 3\abs{V(G)} - 6$.
\end{fact}
When the exact constants are not considered, 
we write $O(\abs{G})$ to denote $O(\abs{V(G)}) = O(\abs{E(G)})$. 

\begin{fact}[Simple-cycle/simple-cut duality~\cite{Whitney}] \label{fact:simple-cycle-simple-cut}
A set of edges forms a simple cycle in a planar embedded graph $G$ iff 
it forms a simple cut in the dual $G^*$.
\end{fact}

Since 
a simple cut in a graph uniquely determines a bipartition of the vertices of the graph, 
a simple cycle in a planar embedded graph $G$ uniquely determines a bipartition of the faces. 

\begin{definition}[Encloses] \label{def:encloses} 
Let $C$ be a simple cycle in a connected planar embedded graph $G$. 
The edges of $C$ form a simple cut $\delta_{G^*}(S)$ for some set $S$ of vertices of $G^*$, 
i.e.~faces of $G$. 
Thus $C$ uniquely determines a bipartition $\set{F_0, F_1}$ of the faces of $G$. 
Let $f_\infty, f$ be faces of $G$. 
We say $C$ {\em encloses}~$f$ with respect to $f_\infty$ if 
exactly one of $f, f_\infty$ is in~$S$. 
For a vertex/edge $x$, 
we say $C$ {\em encloses} $x$ (with respect to $f_\infty$) if 
it encloses some face incident to $x$ 
(encloses {\em strictly} if in addition $x$ is not part of $C$).
\end{definition}

\begin{fact}[\cite{vonStaudt}] 
For any spanning tree $T$ of $G$, 
the set of edges of $G$ not in $T$ form a spanning tree of $G^*$.
\end{fact}

For a spanning tree $T$ of $G$, 
we typically use $T^*$ to denote the spanning tree of $G^*$ 
consisting of the edges not in $T$. 
We often refer to $T^*$ as the {\em cotree} of $T$~\cite{conf/soda/Eppstein03}.

\subsection{Separators}

\begin{definition}
For an assignment $W(\cdot)$ of nonnegative weights to faces, edges, and vertices of $G$, 
we say a simple cycle $C$ is a {\em balanced separator}
if the total weight of faces, edges, and vertices strictly enclosed by~$C$ 
and the total weight not enclosed are each 
at most $3/4$ of the total weight.
\end{definition}
(Traditionally, balance involves a bound of $2/3$. 
We use $3/4$ because it simplifies the presentation.)

One can reduce the case of face/edge/vertex weight to the case of face weight. 
For each vertex or edge, 
remove its weight and add it to an incident face. 
A cycle separator 
that is balanced with respect to the resulting face-weight assignment 
is balanced with respect to the original weight assignment. 
We may therefore assume in cycle-separator algorithms 
that only the faces have weight. 

Let $G$ be a planar embedded graph with face weights. 
Suppose that no face has more than $1/4$ of the total weight.
Lipton and Tarjan~\cite{journals/siamjam/LiptonT79} show that, 
if $G$ is triangulated 
(every face has size at most~3) 
then for any spanning tree $T$ of $G$, 
there is an edge not in $T$ whose fundamental cycle with respect to $T$ 
is a balanced separator.
Goodrich~\cite{journals/jcss/Goodrich95} observed that 
such an edge can be found by looking for an edge-separator in the cotree $T^*$ of $T$.

We modify this approach slightly: 
let $T^*$ be a spanning tree of the planar dual $G^*$ of $G$ 
such that $T^*$ has maximum degree~3. 
Let $T$ be the cotree of $T^*$, 
so $T$ is a spanning tree of $G$. 
Root $T^*$ at an arbitrary vertex of degree one or two. 
Let $v$ be a leafmost vertex of $T^*$ 
such that the descendants of $v$ 
(including $v$) 
comprise more than $3/4$ of the weight. 
Let $\hat e$ be the edge of $T^*$ connecting $v$ to a child with greatest descendant weight.

\begin{lemma} \label{lem:fundamental-cycle-separator} 
The fundamental cycle of $\hat e$ with respect to $T$ 
is a balanced simple cycle separator.
\end{lemma}

\begin{proof} 
Simple algebra shows that
each of the two trees of $T^*\setminus\set{\hat e}$ comprises 
between $1/4$ and $3/4$ of the total weight. 
One of these trees consists of the faces enclosed by~$C$, 
and 
the other consists of the faces not enclosed by~$C$, 
where $C$~is the fundamental cycle of $\hat e$ with respect to~$T$. 
\end{proof}

Miller~\cite{journals/jcss/Miller86} proved the following theorem
(actually, a stronger version).
\begin{theorem*}[Miller~\cite{journals/jcss/Miller86}]
For a planar triangulated biconnected graph $G$ 
with weights such that 
the weight of each face, edge, and vertex comprises at most 2/3 of the total weight, 
there is a simple cycle $C$ of length at most $2\sqrt{2\abs{V(G)}}$ 
such that the total weight strictly enclosed is at most 2/3 
and the total weight not enclosed is at most 2/3.
There is a linear-time algorithm to find such a cycle.
\end{theorem*}
In this paper, 
we do not use Miller's construction. 
We give another construction that, 
with the aid of some auxiliary data structures, 
can be carried out in sublinear time. 
For simplicity of presentation, 
we present a construction that achieves a balance of 3/4 instead of 2/3, 
and refrain from optimizing the various constants that arise in our construction.

\subsection{Divisions}
\label{sec:prel:rdiv}

Frederickson~\cite{journals/siamcomp/Frederickson87} introduced 
the notion of {$r$--divisions}.
Let $\barG$ be an $n$-vertex planar embedded graph.

\begin{definition}
A {\em region} $R$ of $\barG$ is 
an edge-induced subgraph of $\barG$. \end{definition}

\begin{definition} A {\em division} of $\barG$ is 
a collection of regions 
such that each edge is in at least one region. 
A vertex is a {\em boundary vertex} of the division 
if it is in more than one region. 
A division is an $r$--division 
if there are $O(n/r)$ regions, 
and 
each region has at most $r$ vertices 
and $O(\sqrt{r})$ boundary vertices.
\end{definition}

Frederickson's definition does not address the number of holes, 
which is not relevant in his algorithms 
(and in some subsequent algorithms building on his). 
The next definition follows the lines of that of Cabello~\cite{journals/algorithmica/Cabello12} 
but our terminology is slightly different.

\begin{definition}
A {\em natural} face of a region $R$ is 
a face of $R$ that is also a face of $\barG$. 
A {\em hole} of $R$ is a face of $R$ that is not natural.
\end{definition}

\begin{definition} An {\em $r$--division with few holes}
is an $r$--division in which
\begin{itemize}
\item any edge of two regions is on a hole of each of them, and
\item every region has $O(1)$ holes.
\end{itemize}
\end{definition}
This differs from Cabello's definition in that 
his requires only that the {\em average} number of holes per region be $O(1)$. 
We use a stronger requirement 
because some algorithms depend on it.

Our main theorem is as follows.
\begin{theorem} \label{thm:r-division}
For a constant $s$, 
there is a linear-time algorithm that, 
for any biconnected triangulated planar embedded graph $\barG$ 
and any $r\geq s$, 
outputs an $r$--division of $\barG$ with few holes.
\end{theorem}
In fact, our algorithm further guarantees 
that all regions in the $r$--division are connected. 
This property is desirable in some applications (e.g.,~\cite{LNSWN12}). 

\subsection{Recursive Divisions and Decomposition Trees}
\label{sec:prel:recursivediv}
Some algorithms, 
e.g.~the shortest-path algorithm of Henzinger et al.~\cite{journals/jcss/HenzingerKRS97}, 
require that the graph be decomposed into regions 
that are in turn decomposed into regions, 
and so on. 
The Henzinger et al.~algorithm requires roughly $\log^* n$ levels of decomposition, 
so it would take more than linear time 
to find all the different divisions independently. 
We describe a simple decomposition of a planar graph 
that allows one to obtain such recursive divisions in linear time.

\begin{definition}
A {\em decomposition tree} for $\barG$ is 
a rooted tree in which 
each leaf is assigned a region of $\barG$ 
such that each edge of $\barG$ is represented in some region. 
For each node $x$ of the decomposition tree, 
the {\em region $R_x$ corresponding to $x$} is 
the subgraph of $\barG$ that is 
the union of the regions assigned to descendants of~$v$.
\end{definition}

\begin{definition}
A decomposition tree $\mathcal T$ 
{\em admits} an $r$--division with few holes 
if there is a set $S$ of nodes of~$\mathcal T$ 
whose corresponding regions form 
an $r$--division of~$\barG$ with few holes. 
\end{definition}

\begin{theorem} \label{thm:decomposition}
For a constant $s$, 
there is a linear-time algorithm that, 
for any biconnected triangulated planar embedded graph $G$, 
outputs a binary decomposition tree $\mathcal T$ for $G$ 
that admits an $r$--division of $G$ with few holes 
for every $r \geq s$.
\end{theorem}

\begin{definition}
For an exponentially increasing sequence ${\bm r}=(r_1, r_2, \ldots)$ of numbers, 
a {\em recursive ${\bm r}$--division of $G$ with few holes} is 
a decomposition tree for $G$ in which, 
for $i=1,2, \ldots$, 
the nodes at height~$i$ correspond to regions 
that form an $r_i$--division of $G$ with few holes.
\end{definition}
A recursive ${\bm r}$--division is exactly the kind of structure 
that is useful for the linear-time shortest-paths algorithm~\cite{journals/jcss/HenzingerKRS97}, 
although that algorithm does not require one with few holes.

\begin{theorem} \label{thm:recursive-r-division}
There is a linear-time algorithm that, 
given a decomposition tree satisfying the condition of Theorem~\ref{thm:decomposition}, 
and given an increasing sequence $\bm{r}$, 
returns a {\em recursive $\bm{r}$--division of $G$ with few holes}.
\end{theorem}
Note that 
Theorems~\ref{thm:r-division} and~\ref{thm:recursive-r-division} 
follow easily from Theorem~\ref{thm:decomposition}.

\subsection{A Recurrence Relation}
The following lemma, 
whose proof is provided in the appendix, 
is used to analyze several quantities throughout the paper. 
\begin{lemma}\label{lem:recurrence}
Let $\frac{1}{2} \leq \beta < 1$, 
let $c$ be a constant, 
and 
let 
$$T_r(n) \leq \left\{ \begin{array}{cc}
\rho n^\beta + \max_{\{\alpha_i\}} \sum_i T_r(\alpha_i n) & \mbox{if }
n > r \\
0 & \mbox{if } n \leq r \end{array} \right. ,$$
where the maximization is over $\{\alpha_i\}_{i=1}^8$ such that
\begin{enumerate}
\item $\alpha_i \leq 3/4 + c/\sqrt n$ for $i=1, \ldots, 8$, and 
\item $1 \leq \sum_i \alpha_i \leq 1+ c/\sqrt n$.
\end{enumerate}
There exists a constant $s$ such that 
for every $r\geq s$, 
$T_r(n)$ is $O(n/r^{1-\beta})$.
\end{lemma}

\subsection{Dynamic Trees and Euler-tour Trees}
\label{sec:prelim:dtree}

The dynamic-tree data structure of Sleator and Tarjan~\cite{ST83}
represents rooted forests so as 
to support topological operations {\sc Link} and {\sc Cut} 
(removing and adding edges to the forests, respectively), 
and
to support assignments of weights to nodes with operations 
that add a number to the weights of all ancestors of a given node, 
and
that search among the ancestors for a node with minimum or maximum weight. 
Each operation can be performed in $O(\log n)$ amortized time. 
With care, 
given an initial tree, 
a dynamic-tree representation can be constructed so that 
the amortization is over $O(n/\log n)$ operations.

Many extensions and variants have been developed~\cite{ST83,journals/jal/EppsteinITTWY92,Frederickson97,ABHVW04,AHLT05,TW05}.
For example, using the self-adjusting top trees of Tarjan and Werneck~\cite{TW05}, 
one can search among descendants instead of ancestors, 
and 
one can represent embedded trees 
(adapting an idea of Eppstein et al.~\cite{journals/jal/EppsteinITTWY92}). 

The Euler-tour-tree data structure of Henzinger and King~\cite{HK99}
represents an Euler tour of a tree using a balanced binary search tree. 
It supports topological {\sc Link} and {\sc Cut} operations. 
In particular, 
for a vertex~$v$ of a tree, 
any consecutive 
(in the order used by the Euler tour) 
set of children of $v$ 
can be cut in a single operation 
that requires a constant number of binary search tree operations, 
and hence takes $O(\log n)$ time. 
This is useful for 
representing embedded trees, 
by choosing the order in which nodes are visited by the Euler tour 
to be their cyclic order in the embedding. 
The Euler-tour tree can be decorated with additional labels 
to support ancestor and descendant searches in logarithmic time as well.

 \section{Computing a Decomposition Tree}\label{sec:rdiv}
In this section, 
we give a high-level description of 
the algorithm of Theorem~\ref{thm:decomposition} 
for computing a decomposition tree, 
and 
we prove its correctness. 
The input graph $\bar G$ is assumed to be biconnected and triangulated.
It follows that, 
for every region $R$ of $\bar G$, 
every natural face is a triangle.

The algorithm is as follows. 
Given the input graph $\bar G$, 
the algorithm performs some preprocessing 
necessary for the sublinear-time cycle-separator algorithm (Algorithm~\ref{algo:prepro}), 
and 
calls {\sc RecursiveDivide}$(\bar G, 0)$. 
This procedure is given in Algorithm~\ref{algo:rdiv}.

\begin{algorithm}[h!]
\DontPrintSemicolon
 \SetAlgoLined
 let $n =\abs{V(R)}$\;
 \lIf{$n\leq s$}
 {\Return a decomposition tree 
 consisting of a leaf assigned $R$ }\label{line:basecase}\;
 \lIf{$\ell \bmod 3 = 0$}
 {separator chosen below to balance number of vertices}\; 
 \lElseIf{$\ell\bmod 3 = 1$}
 {separator chosen to balance number of boundary vertices}\;
 \lElseIf{$\ell\bmod 3 = 2$}
 {separator chosen to balance number of holes}\;

 let $R'$ be the graph obtained from $R$ as follows: 
 triangulate each hole by 
 placing an artificial vertex in the face 
 and 
 connecting it via artificial edges to 
 all occurrences of vertices on the boundary of the hole \label{line:triangulation}\;

 find a balanced simple-cycle separator $C$ in $R'$ 
 with at most $c\sqrt{n}$ natural vertices \label{line:cyclesep}\;

let $F_0,F_1 $ be 
the sets of natural faces of $R$ 
enclosed and not enclosed by $C$, respectively \label{line:divide}\;
 \For{$i \in \{0,1\}$}
 {let $R_i$ be 
 the region consisting of 
 the edges of faces in $F_i$ and
 the edges of $C$ that are in $R$\label{line:R_i}\;

$\mathcal T_i \leftarrow ${\sc RecursiveDivide}$(R_i,\ell+1)$\;
}
return the decomposition tree $\mathcal T$ consisting of 
a root with left subtree $\mathcal T_0$ and right subtree $\mathcal T_1$

 \caption{{\sc RecursiveDivide}$(R,\ell)$ \label{algo:rdiv}}
\end{algorithm}

This procedure, 
given a connected region $R$ with more than $s$ edges, 
and 
given a recursion-depth parameter~$\ell$, 
first triangulates each hole of $R$ 
by adding an artificial vertex 
and 
attaching it via artificial edges 
to each occurrence of a vertex on the boundary of the hole. 
Let $R'$ be the resulting graph. 
See Figures~\ref{fig:cut-along-cycle1},~\ref{fig:cut-along-cycle2} for an illustration. 
The vertices and edges that are not artificial are {\em natural}. 
Triangulating in this way establishes biconnectivity of $R'$. 

\begin{lemma}\label{lemma:biconnected}
$R'$ is biconnected.
\end{lemma}
\begin{proof}
Let $u$ be a cut vertex of $R$. 
Let $uv$ and $uw$ be two consecutive (in the embedding) edges incident to $u$ 
that belong to different biconnected components. 
Since $vu,uw$ is a $v$-to-$w$ path in $R$, 
the edge $vw$ is not in $R$. 
Hence $uv$ and $uw$ belong to a face $h$ of $R$ 
that is not a triangle, 
hence $h$ is a hole. 
After triangulation, 
$vx_h,x_hw$ is another vertex-disjoint path from $v$ to $w$, 
using the artificial node $x_h$. 
\end{proof}

Next, 
the procedure uses the {\sc SimpleCycleSeparator} procedure
(Section~\ref{sec:cyclesep:algo}) 
to find a simple-cycle separator $C$ 
consisting of at most $c\sqrt{n}$ natural vertices,
where $c$ is a constant 
and 
$n:=\abs{V(R)}$ is the number of vertices of $R$ 
(which is equivalent to the number of natural vertices of $R'$). 
Depending on the current recursion depth $\ell$, 
the cycle separates $R$ in a balanced way 
with respect to either vertices, boundary vertices, or holes. 

Note that 
{\sc SimpleCycleSeparator} is called on $R'$, 
which has more than $n$ vertices 
since it also has some artificial vertices. 
However, 
we show in Lemma~\ref{lem:num-holes} that 
there are at most twelve artificial vertices, 
so even when using a generic algorithm for finding a simple cycle separator, 
the bound of $c\sqrt{n}$ still holds for some choice of $c$ 
(since $n \geq s$). 
In fact, 
our procedure {\sc SimpleCycleSeparator} 
takes into account which vertices are artificial, 
and 
returns a separator consisting of at most $4\sqrt{3n}$ natural vertices, 
and possibly also some artificial vertices.
The number of artificial vertices on the separator 
does not matter in the analysis of {\sc RecursiveDivide}.

The cycle $C$ determines 
a bipartition of the faces of the triangulated graph $R'$, 
which in turn induces a bipartition $(F_0, F_1)$ 
of the natural faces of $R$. 
For $i\in \set{0,1}$, 
let $R_i$ be 
the region consisting of 
the edges bounding the faces in $F_i$, 
together with the edges of $C$ that are in $R$ 
(i.e.~omitting the artificial edges 
added to triangulate the artificial faces). 
See Figures~\ref{fig:cut-along-cycle3},~\ref{fig:cut-along-cycle4} for an illustration.

\begin{figure}
\begin{center}
\subfigure[A schematic diagram of 
a region $R$ with four holes 
(white faces, one of them being the unbounded face). 
\label{fig:cut-along-cycle1}]
{\includegraphics[width=0.4\textwidth]{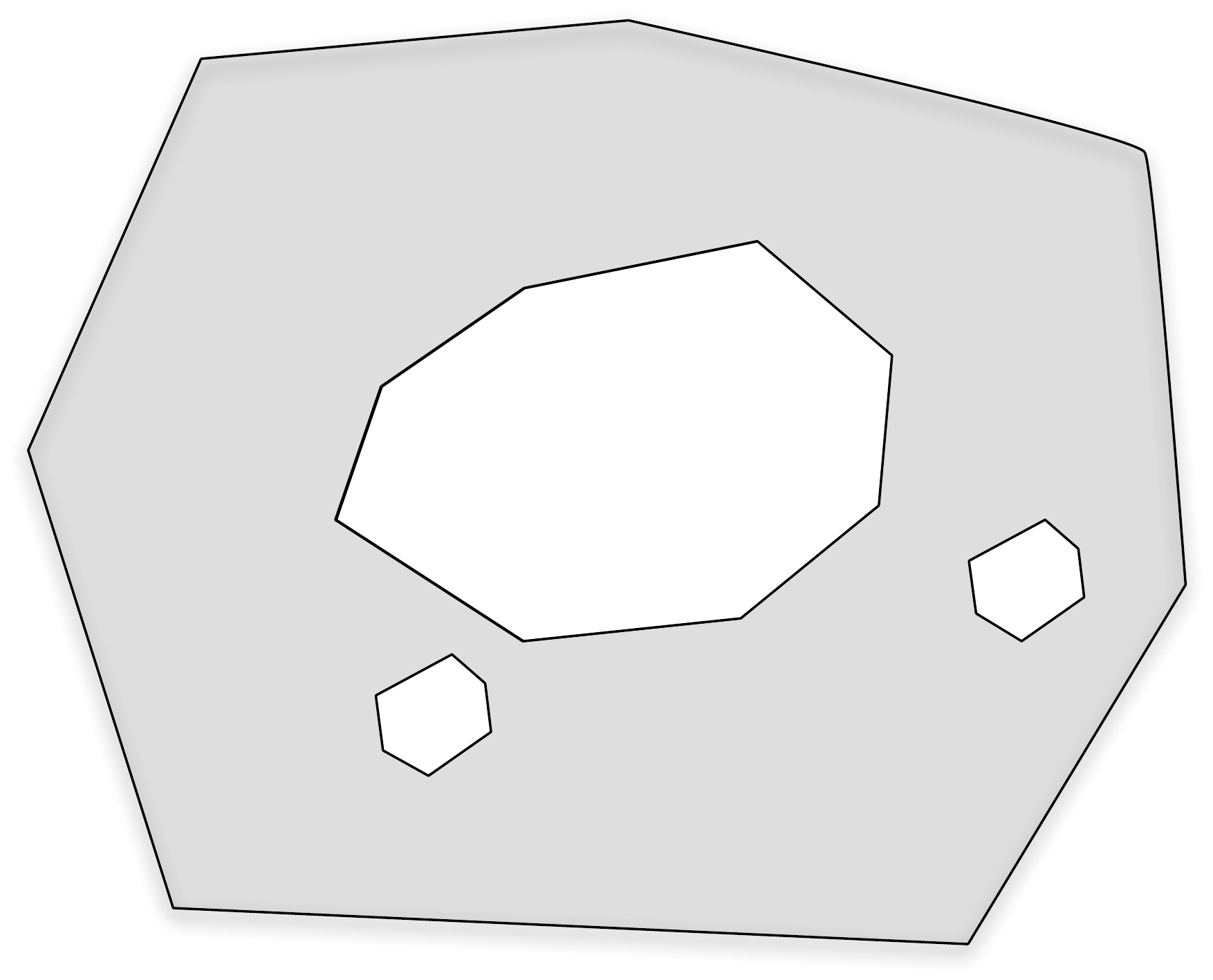}}
\hspace{0.1\textwidth}
\subfigure[The graph $R'$ and a cycle separator $C$ (solid red). 
Artificial triangulation edges are dashed 
(triangulation edges are not shown for the unbounded hole to avoid clutter). 
\label{fig:cut-along-cycle2}]
{\includegraphics[width=0.4\textwidth]{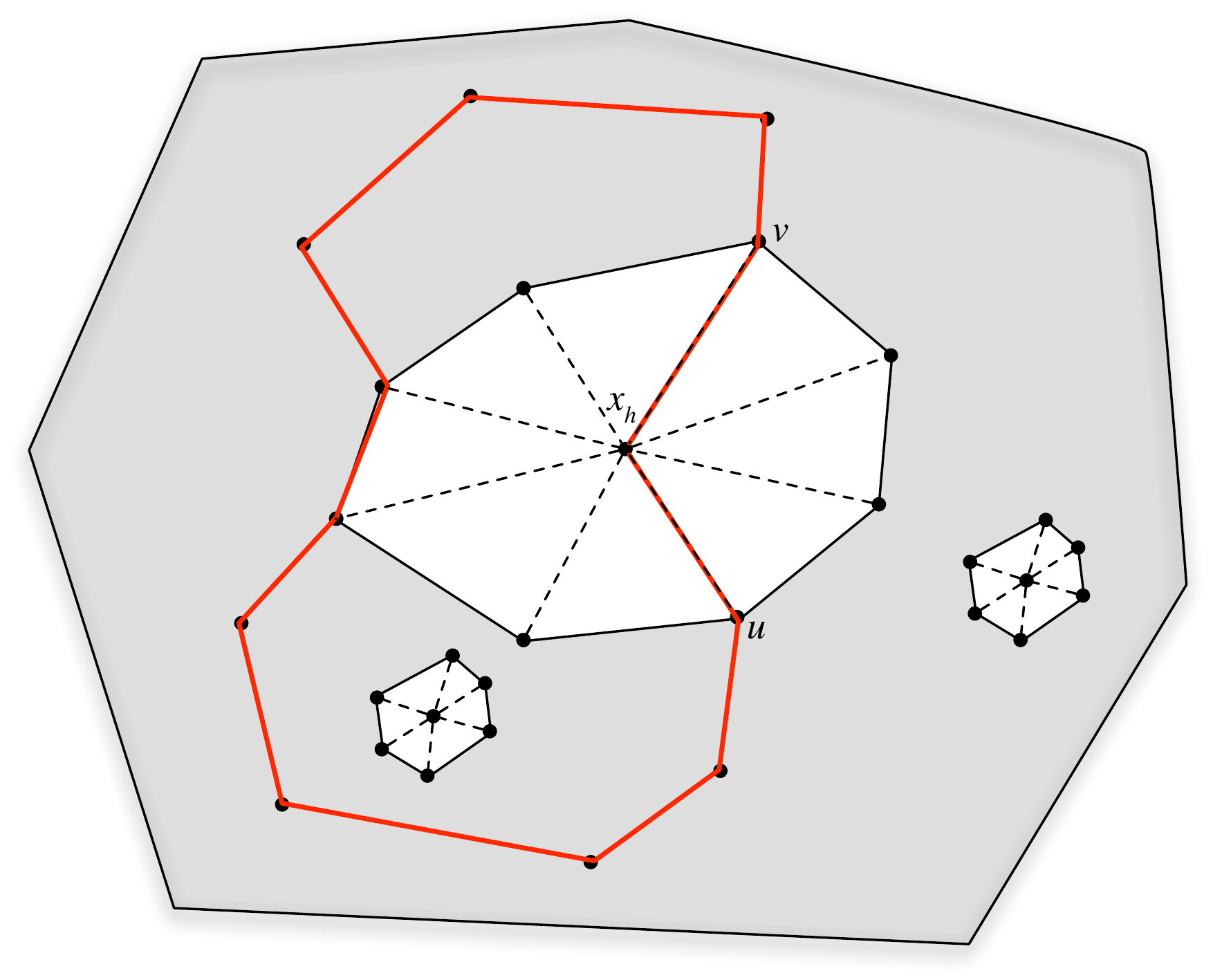}}
\\
\subfigure[The region $R_0$ 
consisting of the edges bounding the faces not enclosed by $C$ 
together with the edges of $C$ that belong to $R$. 
Equivalently, 
$R_0$ is the subgraph of $R'$ not strictly enclosed by $C$
without any artificial edges and vertices. 
$R_0$ has three holes.\label{fig:cut-along-cycle3}]
{\includegraphics[width=0.4\textwidth]{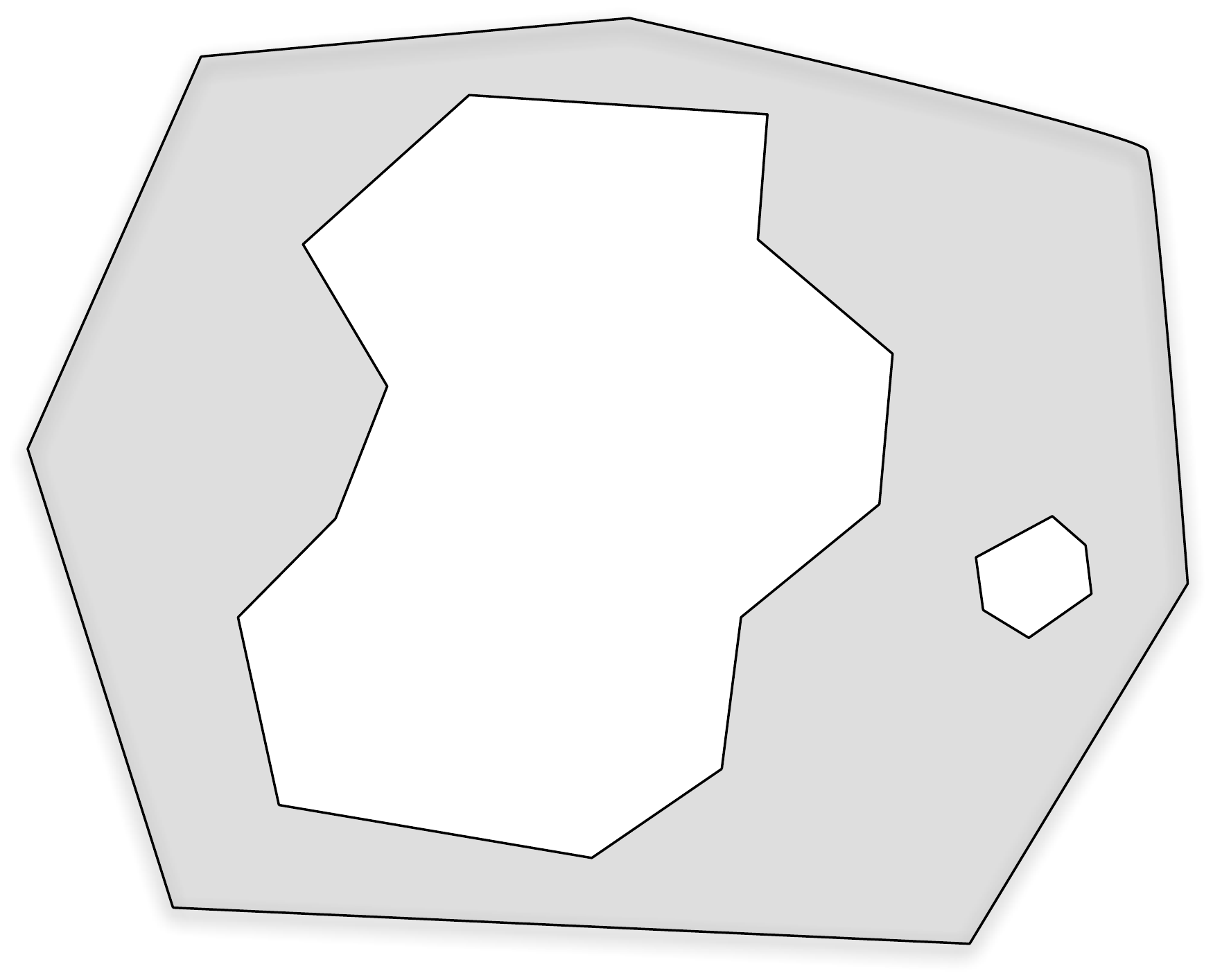}}
\hspace{0.1\textwidth}
\subfigure[The region $R_1$ 
consisting of the edges bounding the faces enclosed by $C$ 
together with the edges of $C$ that belong to $R$. 
Equivalently, 
$R_1$ is the subgraph of $R'$ enclosed by $C$
without any artificial edges and vertices. 
$R_1$ has two holes. 
Note that a hole is not necessarily a simple face.\label{fig:cut-along-cycle4}]
{\includegraphics[width=0.4\textwidth]{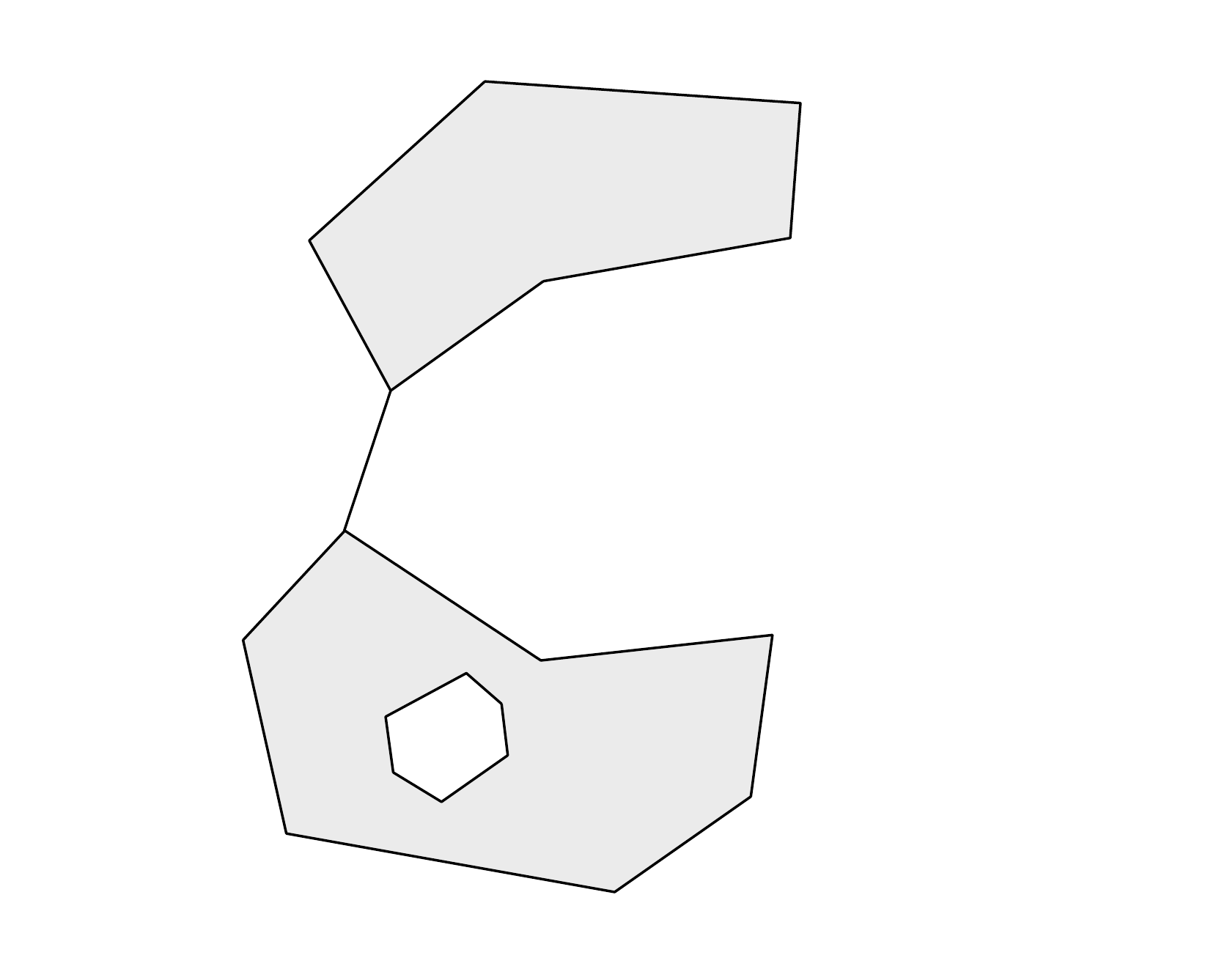}}
\end{center}
\caption{Illustration of triangulating a hole and separating along a cycle.}
\label{fig:cut-along-cycle}
\end{figure}

\begin{lemma} \label{lem:connected}
If $R$ is connected then 
$R_0$ is connected 
and 
$R_1$ is connected.
\end{lemma}
The procedure calls itself recursively on $R_0$ and $R_1$, 
obtaining decomposition trees $\mathcal T_0$ and $\mathcal T_1$, 
respectively. 
The procedure creates 
a new decomposition tree $\mathcal T$ 
by creating a new root corresponding to the region $R$ 
and assigning as its children the roots of $\mathcal T_0$ and $\mathcal T_1$.

\subsection{Number of Holes}

The triangulation step
(Line~\ref{line:triangulation}) 
divides each hole $h$ into a collection of triangle faces. 
We say a hole $h$ is {\em fully enclosed by $C$} if
all these triangle faces are enclosed by $C$ in~$R'$. 

\begin{lemma} \label{lem:enclosing-holes}
Suppose that there are $k$ holes 
that are fully enclosed by $C$. 
Then $R_0$ has $k+1$ holes.
\end{lemma}

\begin{proof} We give an algorithmic proof. 
See Figure~\ref{fig:cut-along-cycle} for an illustration. 
Initialize $R_0'$ to be the graph obtained from $R'$ 
by deleting all edges not enclosed by $C$. 
Then $C$ is the boundary of the infinite face of $R_0'$. 
Consider in turn 
each hole $h$ of $R$ such that 
a nonempty proper subset of $h$'s triangle faces are enclosed by $C$.
For each such face $h$, 
$C$ includes the artificial vertex $x_h$ placed in $h$, 
along with two incident edges $ux_h$ and $x_h v$
where $u$ and $v$ are distinct vertices on the boundary of $h$. 
Deleting all the remaining artificial edges of $h$ 
modifies the boundary of the infinite face 
by replacing $ux_h\ \ x_h v$ 
with a subsequence of the edges 
forming the boundary of $h$. 
In particular, 
deleting these artificial edges 
does not create any new faces. 

Finally, 
for each hole $h$ that is fully enclosed by $C$, 
delete the artificial edges of $h$, 
turning $h$ into a face of $R_0'$. 
The resulting graph is $R_0$, 
whose holes are the holes of $R$ that were fully enclosed by $C$, 
together with the infinite face of $R_0$.
\end{proof}

If the recursion depth $\bmod$~3 is 2, 
Line~\ref{line:cyclesep} of {\sc RecursiveDivide} 
must select a simple cycle in $R'$ 
that is balanced with respect to the number of holes.
To achieve this, 
for each hole $h$ of $R$, 
the algorithm assigns weight~1 
to one of the triangles resulting from triangulating $h$ in Line~\ref{line:triangulation}, 
and weight~0 to all other faces. 
Then 
the algorithm finds a cycle $C$ 
that is balanced with respect to these face-weights.

\begin{lemma} \label{lem:num-holes} For any region 
created by {\sc RecursiveDivide}, 
the number of holes is at most twelve.
\end{lemma}

\begin{proof} By induction on the recursion depth $\ell$, 
\begin{enumerate}
\item[(A)] if $\ell \bmod 3 = 0$ then $R$ has at most ten holes;
\item[(B)] if $\ell \bmod 3 = 1$ then $R$ has at most eleven holes;
\item[(C)] if $\ell \bmod 3 = 2$ then $R$ has at most twelve holes.
\end{enumerate}
For $\ell=0$, 
there are no holes.
Assume (C) holds for $\ell$. 
Since the cycle separator $C$ encloses at most 3/4 of the weight, 
it fully encloses at most nine holes. 
By Lemma~\ref{lem:enclosing-holes}, 
$R_0$ has at most ten holes. 
The symmetric argument applies to $R_1$. 
Thus (A) holds for $\ell+1$.

Similarly, 
by Lemma~\ref{lem:enclosing-holes}, 
if (A) holds for $\ell$ 
then (B) holds for $\ell+1$, 
and 
if (B) holds for $\ell$ 
then (C) holds for $\ell+1$. 
\end{proof}

\subsection{Number of Vertices and Boundary Vertices}
\label{sec:num-of-vertices}
If the recursion depth $\bmod$~3 is 0, 
Line~\ref{line:cyclesep} of {\sc RecursiveDivide} 
selects a simple cycle in $R'$ 
that is balanced with respect to the number of natural vertices. 
To achieve this, 
for each natural vertex $v$,
the algorithm selects an adjacent face in $R'$, 
dedicated to carry $v$'s weight. 
The weight of each face is defined to be 
the number of vertices for which that face was selected.
Since each face in $R'$ is a triangle, 
every weight is an integer between~0 and~3. 
A cycle $C$ is then chosen 
that is balanced with respect to these face-weights. 

If the recursion depth $\bmod$~3 is 1, 
the cycle must be balanced with respect to 
the number of boundary vertices. 
For each boundary vertex, 
the algorithm selects an incident face; 
the algorithm then proceeds as above.

In either case, 
the total weight enclosed by the cycle $C$ is 
an upper bound on the number of vertices 
(natural or boundary) 
strictly enclosed by $C$. 
Thus at most 3/4 of the vertices 
(natural or boundary) of $R'$ 
are strictly enclosed by $C$ in $R'$. 
Similarly,
at most 3/4 of the vertices are not enclosed by $C$ in $R'$.

The vertices of $R_0$ are the natural vertices of $R'$ 
enclosed by $C$ 
(including the natural vertices on $C$, 
which number at most $c\sqrt{\abs{V(R)}}$), 
and 
the vertices of $R_1$ are 
the natural vertices of $R'$ not strictly enclosed by $C$. 
Let $n:=\abs{V(R)}$ and, 
for $i\in\set{0,1}$,
let $n_i:=\abs{V(R_i)}$. 
We obtain
\begin{equation} \label{eq:total-vertices}
n_0+n_1 \leq n + c\sqrt{n}.
\end{equation}
Moreover, 
if the recursion depth $\bmod$~3 is 0, 
then
\begin{equation} \label{eq:vertices-per-piece}
\max\set{n_0,n_1} \leq \frac{3}{4} n + c\sqrt{n}.
\end{equation}
Similarly, 
let $b$ be the number of boundary vertices of $R$, 
and, 
for $i\in\set{0,1}$, 
let $b_i$ be the number of boundary vertices of $R_i$. 
We obtain
\begin{equation} \label{eq:total-boundary-vertices}
b_0+b_1 \leq b + c\sqrt{n}. 
\end{equation}
Moreover, 
if the recursion depth $\bmod$~3 is 1, 
then
\begin{equation} \label{eq:boundary-vertices-per-piece}
\max \set{b_0, b_1} \leq \frac{3}{4} b + c \sqrt{n}.
\end{equation}

\subsection{Admitting an $r$--division}
Let $N$ be 
the number of vertices in the original input graph $\barG$.
Consider the decomposition tree $\mathcal T$ of $\barG$ 
produced by {\sc RecursiveDivide}. 
Each node $x$
corresponds to a region $R_x$. 
We define $n(x) := \abs{V(R_x)}$.
In this section 
we show that, 
for any given $r \geq s$, 
$\mathcal T$ admits an $r$--division of $\barG$. 
We adapt the two-phase analysis of Frederickson~\cite{journals/siamcomp/Frederickson87}.
In the first phase (Lemma~\ref{lem:phase1}), 
we identify a set of $O(N/r)$ regions 
for which the average number of boundary vertices is $O(\sqrt{r})$. 
However, 
some of the individual regions in this set 
might have too many boundary vertices 
(since the number of vertices and boundary vertices 
do not necessarily decrease at the same rate).
We show that each such region 
can be replaced with smaller regions in $\mathcal T$ 
so that every region 
has $O(\sqrt{r})$ boundary vertices, 
and
the total number of regions remains $O(N/r)$ 
(Lemma~\ref{lem:phase2}).

For a node $x$ of $\mathcal T$ 
and 
a set $S$ of descendants of $x$
such that no node in $S$ is an ancestor of any other,
define $L(x, S) := -n(x) + \sum_{y\in S} n(y)$. 
Roughly speaking, 
$L(x,S)$ counts the number of new boundary nodes 
with multiplicities 
when replacing $x$ by all regions in $S$. 
\begin{lemma} \label{L-additive}
If $y\in S_1$ 
then $L(x,S_1) + L(y,S_2)=L(x, S_1\cup S_2-\set{y})$.
\end{lemma}

\begin{proof} Each node $z\in S_2$ contributes 
$n(z)$ to both sides of the equation. 
Since $y$ is in $S_1$, 
it contributes $n(y)$ to $L(x,S_1)$ 
(the first term on the left). 
It contributes $-n(y)$ to $L(y,S_2)$ 
(the second term on the left), 
and nothing to the righ-hand side. 
Every other node $y'\in S_1$ 
contributes $n(y')$ to both sides.
\end{proof}

If the children of $x$ are $x_0$ and $x_1$, 
Equation~\eqref{eq:total-vertices} implies 
\begin{equation} \label{eq:L-children}
L(x, \set{x_0,x_1}) \leq c\sqrt{n(x)}.
\end{equation}
Fix $r$ and 
let $S_r$ be 
the set of nodes $y$ of $\mathcal T$ 
such that $y$'s region has no more than $r$ vertices 
but the region of $y$'s parent has more than $r$ vertices. 
Note that no node in $S_r$ is an ancestor of any other. 
Let $\hat x$ be the root of $\mathcal T$. 
\begin{lemma}[Total Number of Boundary Vertices]\label{lem:phase1} 
There are constants $s$ and $\gamma$, 
depending on $c$, 
such that,
for any $r > s$, 
$L(\hat x, S_r) \leq \frac{\gamma N}{\sqrt{r}}$.
\end{lemma}
\begin{proof}
For each node $x$ of $\mathcal T$, 
let $\ell(x)$ denote 
the depth of $x$ in $\mathcal T$.
For a node $x$ of $\mathcal T$, 
let $S_r(x)$ denote the set of descendants $y$ 
such that $y$'s region has no more than $r$ vertices
but the region of $y$'s parent has more than $r$ vertices. 
For an integer $n>r$, 
define $B_r(n)$ to be 
$\max \set{L(x, S_r(x))\ : \ x \in \mathcal T, n(x) \leq n, \ell(x)=0 \bmod 3}$.

 The lemma is a consequence of the following recurrence, 
 which is justified below.
 \begin{equation} \label{eq:recurrence}
 B_r(n) \leq 7c\sqrt n + \max_{\{\alpha_i\}} \sum_i B_r(\alpha_i n),
 \end{equation}
 where the $\max$ is over $\alpha_1, \ldots, \alpha_8$ 
 between 0 and $\frac{3}{4}+\frac{c}{\sqrt{n}}$
such that
\begin{equation} \label{eq:upper-bound-on-sum}
 \sum_{i=1}^8 \alpha_i \leq 1+7c/\sqrt{n}
\end{equation}
and
\begin{equation} \label{eq:lower-bound-on-sum}
\sum_{i=1}^8 \alpha_i \geq 1.
\end{equation}

 Let $x$ be a node such that $r < \abs{n(x)} \leq n$.
 Let $y_1, \ldots, y_k$ be the rootmost descendants $y$ of $x$ 
 such that $\abs{n(y)}\leq r$ or $\ell(y) \bmod 3=0$. 
 Note that $k\leq 8$.
Repeated application of Lemma~\ref{L-additive} 
and Equation~\eqref{eq:L-children} 
yield 
\begin{equation} \label{eq:L-x}
L(x,\set{y_1, \ldots, y_k}) \leq (k-1) c\sqrt{n},
\end{equation}
which implies that 
\begin{equation} \label{eq:sum-of-sizes-of-yi}
\sum_{i=1}^k n(y_i) \leq n + (k-1) c\sqrt{n}. 
\end{equation}
 For $1\leq i \leq k$, 
 define $\alpha_i$ 
 so that $\alpha_i n =n(y_i)$. 
 For $k < i \leq 8$, 
 define $\alpha_i=0$.
Equation~\eqref{eq:sum-of-sizes-of-yi} then implies 
Equation~\eqref{eq:upper-bound-on-sum}.
Since each vertex of $R_x$ occurs in at least one $R_{y_i}$,
Equation~\eqref{eq:lower-bound-on-sum} holds.
Let $x_0, x_1$ be the children of $x$.
For each $1\leq i \leq k$, 
$R_{y_i}$ is a subgraph of $R_{x_0}$ or $R_{x_1}$, 
so $\alpha_i n \leq \frac{3}{4} n + c \sqrt{n}$, 
which implies that $\alpha_i \leq \frac{3}{4} + c/\sqrt{n}$.

By definition of $B_r(\cdot)$, 
$L(y_i, S_r(y_i)) \leq B_r(\alpha_i n)$, 
so we have
\begin{equation} \label{eq:sum_of_descendant-sizes-of-yi}
\sum_{i=1}^8 L(y_i, S_r(y_i)) \leq \sum_{i=1}^8 B_r(\alpha_i n).
\end{equation}
By Lemma~\ref{L-additive}, 
$$L(x, S_r(x)) \leq L(x, \set{y_1, \ldots, y_k}) + \sum_{i=1}^k L(y_i,S_r(y_i)),$$
so Equation~\eqref{eq:sum_of_descendant-sizes-of-yi} and 
Equation~\eqref{eq:L-x} together 
justify the recurrence relation~\eqref{eq:recurrence}.
By Lemma~\ref{lem:recurrence}, 
this recurrence relation yields 
that $B_r(n)$ is $O(n/\sqrt r)$.
\end{proof}

Lemma~\ref{lem:phase1} implies 
that $\sum_{x\in S_r} n(x)$ is $O(N)$, 
since the regions in $S_r$ are disjoint 
{\em except} for boundary vertices, 
of which there are at most $O(N/\sqrt{r})$. 
For each parent $y$ of a node $x\in S_r$, 
the corresponding region~$R_y$ has more than $r$ vertices, 
so the number of such parents is $O(N/r)$, 
so $\abs{S_r}$ is $O(N/r)$.
Let $c'$ be a constant to be determined. 
For a node $x$, 
let $S'_r(x)$ denote the set of rootmost descendants $y$ of $x$ 
(where $x$ is a descendant of itself) 
such that $R_y$ has at most $c' \sqrt{r}$ boundary vertices. 
Let $S'_r = \bigcup_{x\in S_r} \set{S'_r(x)}$.

\begin{lemma}\label{lem:phase2}
The regions $\set{R_y \ :\ y \in S'_r}$ form 
an $r$--division with a constant number of holes per region.
\end{lemma}
\begin{proof}
It follows from the definition of $S'_r$ 
that each region in this set has at most $r$ vertices and 
at most $c'\sqrt{r}$ boundary vertices. 
It follows by induction from Lemma~\ref{lem:connected} 
that each of these regions is connected. 
It follows from Lemma~\ref{lem:num-holes} 
that each region has at most twelve holes. 
It remains to show that $\abs{S'_r}$ is $O(N/r)$.

For a node $x$ of $\mathcal T$, 
let $b(x)$ denote the number of boundary vertices of $R_x$. 
Lemma~\ref{lem:phase1} implies that 
\begin{equation} \label{eq:phase-1-consequence}
\sum\limits_{x\in S_r} b(x) \leq \frac{N}{\sqrt{r}}.
\end{equation}
We claim that, 
for every node $x\in S_r$, 
\begin{equation} \label{eq:phase-2-claim}
\abs{S'_r(x)} \leq \max \set{1, \frac{b(x)}{c\sqrt{r}} - 12}
\end{equation}
Summing over $x\in S_r$ 
and using Equation~\eqref{eq:phase-1-consequence}
then proves that $\abs{S'_r}$ is $O(N/r)$.

We set $c'=40c$. 
Proof of Equation~\eqref{eq:phase-2-claim} is by induction. 
If $b(x)\leq c'\sqrt{r}$ 
then $S'_r(x)=\set{x}$, 
so the claim holds. 
Assume therefore that $b(x) > c'\sqrt{r}$. 
Let $y_1,\ldots, y_k$ be 
the rootmost descendants $y$ of $x$ 
such that $b(y) \leq c'\sqrt{r}$ 
or $\ell(y)-\ell(x)=3$, 
ordered such that
$\abs{S'_r(y_1)} \geq \abs{S'_r(y_2)} \geq \cdots \geq \abs{S'_r(y_8)}$. 
Let $q$ be the cardinality of $\set{i\ : \abs{S'_r(y_i)}>1}$.

Case 0: $q=0$. 
In this case, 
$\abs{S'_r(x)}\leq 8$, 
and 
$\frac{b(x)}{c\sqrt{r}} - 12 > 40-12 \geq 8$.

Case 1: $q=1$. 
For some ancestor $y$ of $y_1$ 
that is a descendant of $x$, 
the separator chosen for $R_y$ 
is balanced in terms of boundary vertices. 
It follows that 
$b(y_1) \leq \frac{3}{4} b(x) + 3c\sqrt{r}$. 
By the inductive hypothesis, 
$\abs{S'_r(y_1)}\leq \frac{b(y_1)}{c\sqrt{r}} - 12$, 
so 
\begin{eqnarray*}
\abs{S'_r(x)} & = & \abs{S'_r(y_1)} + k-1\\
 & \leq & \frac{b(y_1)}{c\sqrt{r}} - 12 + 7\\
 & \leq & \frac{\frac{3}{4} b(x)+3c\sqrt{r}}{c\sqrt{r}} - 12 + 7\\
 & \leq & \frac{b(x)}{c\sqrt{r}} - \frac{b(x)}{4c\sqrt{r}} + 3 - 12 + 7\\
 & \leq & \frac{b(x)}{c\sqrt{r}} - 12,
\end{eqnarray*}
because $b(x) > c'\sqrt{r} = 40c\sqrt{r}$.

Case 2: $q=2$. 
In this case, 
$b(y_1)+b(y_2) \leq b(x) + 6c\sqrt{r}$. 
Using the inductive hypothesis on $y_1$ and $y_2$, 
we have
\begin{eqnarray*}
\abs{S'_r(x)} & = & \abs{S'_r(y_1)} +\abs{S'_r(y_2)} + k-2\\
 & \leq & \frac{b(y_1)}{c\sqrt{r}} - 12 + \frac{b(y_2)}{c\sqrt{r}} - 12 - 6\\
 & \leq & \frac{b(x)+6c\sqrt{r}}{c\sqrt{r}} - 12 - 12 -6\\
& \leq & \frac{b(x)}{c\sqrt{r}} - 12.
\end{eqnarray*}

Case 3: $q>2$. 
This case is similar to Case 2.
\end{proof}

\section{A Simple-Cycle Separator Algorithm}\label{sec:cyclesep}
In this section we present our cycle separator algorithm. 
As a one-shot algorithm, the input is
a simple biconnected graph $G$ with $m$ edges and face weights such that no face weighs more
than $3/4$ the total weight and no face consists of more than 3 edges.
The algorithm outputs a simple cycle $C$ in $G$, 
such that neither the total weight strictly enclosed by $C$ nor the total
weight not enclosed by $C$ exceeds $3/4$ of the total weight.
The length of $C$ is guaranteed to be at most $4\sqrt{\abs{E(G)}}$. Since
$G$ is a simple graph, this implies a bound of $4\sqrt{3\abs{V(G)}}$
on the length of the cycle. Better constants can be achieved, at the
cost of complicating the algorithm and the analysis, which we avoid for
the sake of simplicity and ease of presentation.

In a similar manner, we aim for $3/4$--balance to simplify the presentation. A balance of $2/3$
can be achieved.
Note that handling face-weights can be used to handle vertex-weights
and edge-weights as well; simply assign the weight of every vertex and
every edge to an arbitrary incident face. This works as long as the
resulting face weights satisfy the requirement that no single face
weighs more than $3/4$ of the total weight.

The cycle separator algorithm consists of a preprocessing step, which
runs in linear time and computes certain auxiliary data structures 
used by the main procedure, {\sc SimpleCycleSeparator}, which performs
the computation. These data structures can be represented so that {\sc
 SimpleCycleSeparator} takes sublinear time. One auxiliary data structure
is a tree $\mathcal K$, which is called the {\em component
 tree}. The tree $\mathcal K$ captures the
structural connectivity of dual BFS components of $G$. The dual BFS
components satisfy a certain disjointness property (see
Lemma~\ref{lem:disjoint-cycles}). The other
auxiliary data structures are a
spanning tree $T$ of $G$ and its cotree~$T^*$. The spanning tree $T$
satisfies a certain monotonicity property (see
Lemma~\ref{lem:Tparent}). The disjointness and monotonicity properties guarantee that the
length of the cycle separator output by {\sc SimpleCycleSeparator} is $4\sqrt{\abs{E(G)}}$.

Our algorithm for constructing the decomposition tree $\mathcal T$
invokes {\sc SimpleCycleSeparator} multiple times, on various regions
of $\barG$. It computes the auxiliary data structures $\mathcal K$, $T$,
and $T^*$ once, for the input graph $\barG$, and efficiently updates their
representation when separating a region into two regions. However, the
disjointness and monotonicity properties mentioned above are slightly weaker (see
Invariant~\ref{inv:disjoint-cycles} and Invariant~\ref{inv:Tparent});
they only apply to natural edges (disjointness) and to natural vertices (monotonicity). 
One implication is that the number of natural vertices on the cycle
separator produced by {\sc SimpleCycleSeparator} is bounded in terms
of the (squared root of the) number of natural vertices in the input graph. However, the cycle
may also consist of some artificial vertices. As argued in
Section~\ref{sec:rdiv}, artificial nodes on $C$ do not affect the analysis.

\subsection{Levels and Level Components}
\label{sec:components}

We define levels with respect to an arbitrarily chosen face
$f_\infty$, which we designate as the infinite face.

\begin{definition} The {\em level} $\dualfacelevel(f)$ of
a face $f$ is the minimum number of edges on a
$f_\infty$--to--$f$--path in $G^*$. We use $L^F_i$ to denote the
faces having level $i$, and we use $L^F_{\geqslant i}$ denote the set
of faces $f$ having level at least $i$.
\end{definition}

\begin{definition} For an integer $i\geq 0$, a connected component of the
 subgraph of $G^*$ induced by $L^F_{\geqslant i}$ is called a {\em
 level-$i$ component}, or, for unspecified $i$, a {\em
 level component}. We use $\mathcal K_{\geqslant i}$ to denote the
 set of level-$i$ components. A level-$i$ component $K$ is said to
 have level~$i$, and we denote its level by
 $\dualcomponentlevel(K)$. A {\em non-root} level component is a
 level component whose level is not zero.
 The set of vertices of $G^*$ (faces of $G$)
 belonging to $K$ is denoted $F(K)$.
\end{definition}
Note that we use $K$ (not $K^*$) to denote a level component even
though it is a connected component of a subgraph of the planar dual.
$K$ should be thought of as a set of faces. Thus we can refer to it as
a subgraph of $G^*$ or of $G$. In the former case $K$ is the subgraph
of $G^*$ induced by the faces in the set $K$. In the latter case $K$ is the subgraph
of $G$ induced by the edges that belong to faces in the set $K$.
\begin{lemma} \label{lem:outside-connected} For any non-root level
 component, the subgraph of $G^*$ consisting of faces not in $F(K)$
 is connected.
\end{lemma}

\begin{corollary} \label{cor:primal-cycle}
For any non-root level component $K$, the edges of $\delta_{G^*}(F(K))$ form a simple cycle in the primal~$G$.
\end{corollary}

In view of Corollary~\ref{cor:primal-cycle}, for any non-root level
component $K$, we use $X(K)$ to denote the simple cycle in the primal
$G$ consisting of the edges of $\delta_{G^*}(F(K))$. We refer to
$X(K)$ as the {\em bounding cycle} of $K$ since, when viewed as a subgraph
of $G$, $K$ is exactly the subgraph enclosed by $X(K)$. 

\begin{lemma}\label{lem:disjoint-cycles}
Let $K$ and $K'$ be two distinct components. $X(K)$ and $X(K')$ are 
edge-disjoint.
\end{lemma}
\begin{proof}
Let $i$ be the level of $K$. The edges of $X(K)$ are edges of
$\delta_{G^*}(F(K))$. Therefore, as edges of $G^*$, they have
one endpoint in $K$ and one in a level-$(i-1)$ face. If an edge of
$\delta_{G^*}(F(K'))$ has an endpoint in $K$ then $K \neq K'$ implies
the level of $K'$ is at least $i+1$, so it cannot be an edge of $X(K)$. If, on the other hand, an edge
of $\delta_{G^*}(F(K'))$ has an endpoint in level $i-1$ then $K \neq
K'$ implies the other endpoint is not in $K$.
\end{proof}

The following definition is illustrated in Figure~\ref{fig:comptree}.
\begin{definition} \label{def:component-tree}
The {\em component tree} $\mathcal K$ is the
 rooted tree whose nodes are the level components and in which $K$ is
 an ancestor of $K'$ if the faces of $K$ include the faces of $K'$.
\end{definition}
The root of the component tree is the unique level-0 component
consisting of all of $G^*$.

\begin{figure}%[a]
\begin{center}
\subfigure[A triangulated graph (blue vertices and edges) along with a
dual BFS tree (in red, faces at different BFS levels are indicated by
different shapes). The (primal) components $\mathcal K_{\geqslant k}$ are
indicated as shaded subgraphs. The deeper the level of a component,
the darker its shade is.]
{\includegraphics[width=0.5\textwidth]{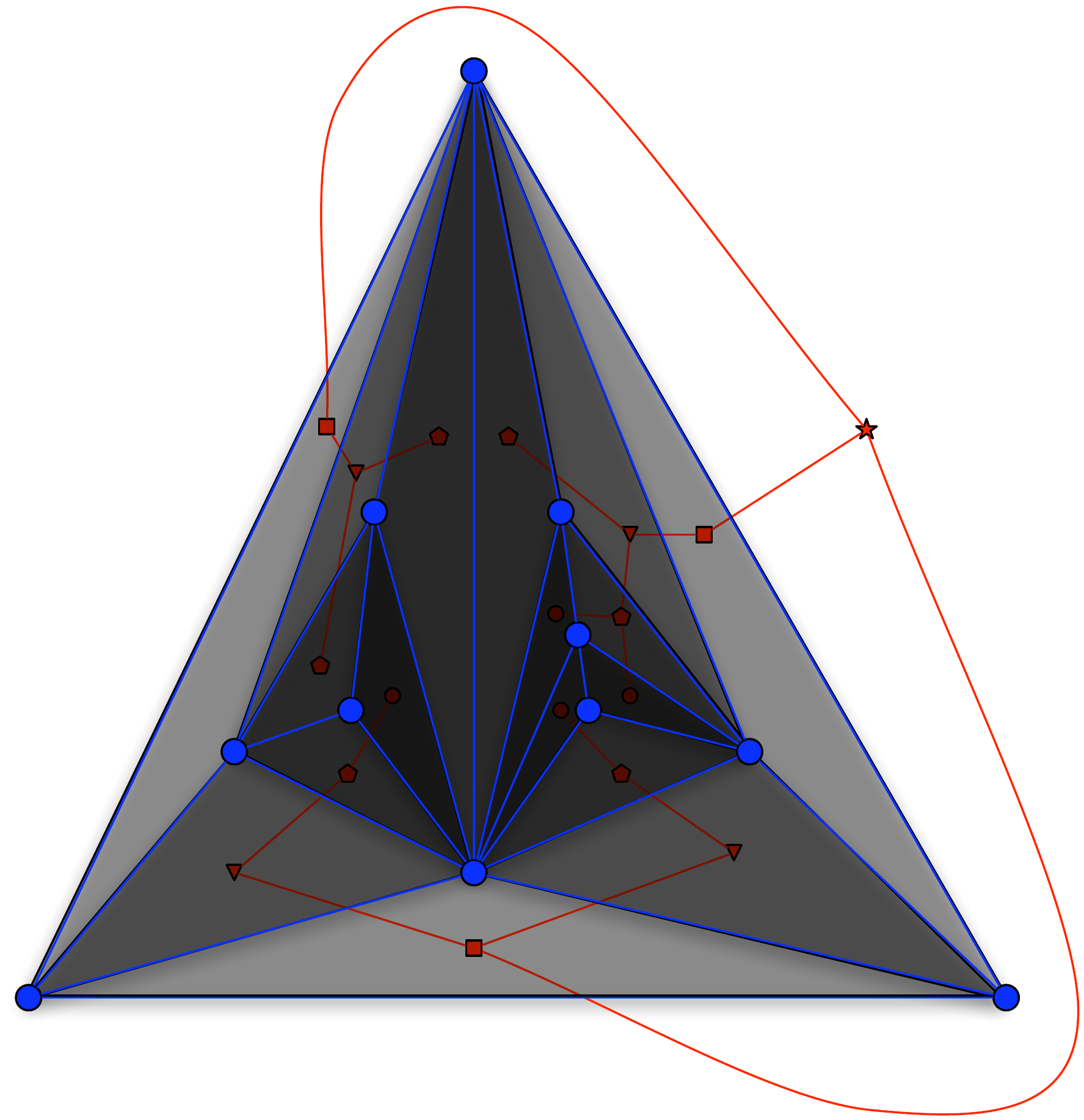}}
\hspace{0.1\textwidth}
\subfigure[The corresponding component tree $\mathcal K$. Each node of
the component tree corresponds to a connected component of $G$.]
{\includegraphics[width=0.22\textwidth]{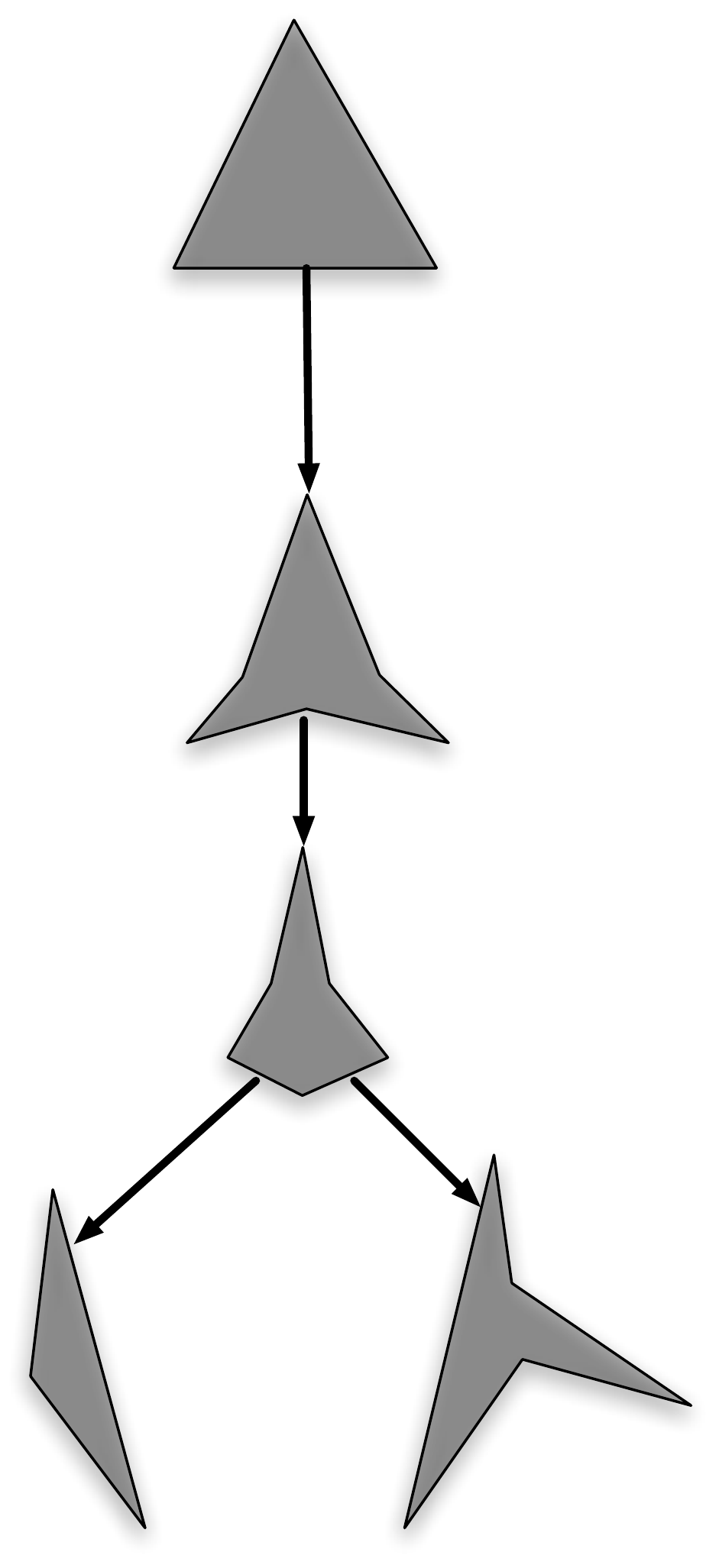}}
\end{center}
\caption{Illustration of the component tree $\mathcal K$.}
\label{fig:comptree}
\end{figure}

\begin{definition}
An edge $f f'$ of
$G^*$ has level $i$ if $f$ has level~$(i-1)$ and $f'$ has level~$i$. We
write $\dualedgelevel(f f')$ for the level of $ff'$. We
use $L^E_i$ to denote the set of edges of level $i$. 
\end{definition}
Note that not every edge
of $G^*$ is assigned a level. 
\begin{definition} Let
$L^V_i$ denote the set of vertices of the primal graph $G$ that are
endpoints in the primal graph $G$ of edges in $L^E_i$.
\end{definition}
 Note that a vertex of $G$ can be an endpoint of
two edges at different levels $i$ and $j$, so $L^V_i$ and $L^V_j$ are
not necessarily disjoint.

\begin{definition} The {\em level} $\dualnodelevel(v)$ of a primal vertex~$v$ 
is defined to be $\min_f\dualfacelevel(f)$ over all
faces $f$ incident to~$v$.
\end{definition}
Note that $L^V_i$ is {\em not} the set of vertices with level~$i$.

\subsection{The Primal Tree}
\label{sec:cyclesep:mirrortree}
The algorithm maintains a primal spanning tree $T$. We start by
describing the initial value of $T$ and its properties.

\begin{lemma}\label{lem:Tparent}
There exists a spanning tree $\bar T$ such that, for any vertex $u$, 
$\dualnodelevel(\parent_{\bar T}(u)) < \dualnodelevel(u)$. 
$\bar T$ can be computed in linear time. 
\end{lemma}
\begin{proof}
For a primal vertex~$u$ 
with $\dualnodelevel(u) = i$, let $f$ be a level-$i$ face to which
$u$ is incident (ties are broken arbitrarily, but consistently). 
Let $v,w$ be the other two vertices of that face~$f$. The parent of $u$ in
$\bar T$, denoted by $\parent_{\bar T}(u)$, is the vertex in $\set{v,w}$ with the
smaller level (again, breaking ties arbitrarily and consistently).

Since $\dualnodelevel(u) = i$, $u$ is not incident to a level-$(i-1)$
face. Hence $v,w$ must be adjacent to a level-$(i-1)$ face $f'$.
Therefore, $\dualnodelevel(\parent_{\bar T}(u)) \leq i-1 <
\dualnodelevel(u) = i$. 

To complete the definition of $\bar T$, we choose an arbitrary vertex~$r$ incident to $f_\infty$ to be the
root of $\bar T$ by assigning it to be the parent of the two remaining
vertices at level $i=0$. 
For convenience we set the level of the root vertex to $-1$. 
\end{proof}

\subsection{The Preprocessing Step}
To compute a simple cycle separator one has to first compute the
component tree $\mathcal K$, the spanning tree $T$, and its cotree
$T^*$. This is done by the preprocessing step
(Algorithm~\ref{algo:prepro}), which runs in linear time.

\begin{algorithm}[htb]
\DontPrintSemicolon
 \SetAlgoLined
 choose an arbitrary face as $f_\infty$\; 
 compute face, edge, and vertex levels $\dualfacelevel(\cdot)$, $\dualedgelevel(\cdot)$, and $\dualnodelevel(\cdot)$, respectively\;
 compute the component tree $\mathcal K$\;
 initialize $T$ to be the tree $\bar T$ as defined in Lemma~\ref{lem:Tparent}\;
 initialize $T^*$ to be the cotree of $T$\;
 \Return $(\mathcal K,T,T^*)$
 \caption{{\sc Preprocess}$(G)$\label{algo:prepro}}
\end{algorithm}

The efficient implementation of {\sc RecursiveDivide}, which computes the
decomposition tree $\mathcal T$ in linear time,
recursively separates regions of $\barG$. It
maintains the component tree $\mathcal K$, spanning tree $T$, and
cotree $T^*$ of the currently handled region $R$
throughout the recursive calls. This is described in detail in Section~\ref{sec:maint}.
$T$ is initialized to be the tree $\bar T$ of $\barG$. {\sc
 RecursiveDivide} maintains the following invariants:
\begin{invariant}\label{inv:disjoint-cycles}
Let $K$ and $K'$ be two distinct components. $X(K)$ and $X(K')$ do not
share natural edges.
\end{invariant}
\begin{invariant}\label{inv:Tparent}
$\dualnodelevel(v) < \dualnodelevel(u)$ for any two natural vertices $u$ and $v$ of $R$ such that $v$ is an
ancestor of $u$ in $T$. 
\end{invariant}
Since $\barG$ has only natural vertices and edges,
Lemma~\ref{lem:disjoint-cycles} and Lemma~\ref{lem:Tparent}
show that the invariants initially hold for $R=\barG$.

\subsection{Computing a Simple Cycle Separator} \label{sec:cyclesep:algo}

We first provide an intuitive description of the separator algorithm
{\sc SimpleCycleSeparator}
(Algorithm~\ref{algo:cycle}). 
Let $m$ denote the number of natural edges in~$G$. 

\begin{enumerate}
\item The algorithm computes a balanced fundamental cycle $\tilde C$
(Lines~\ref{line:fundcyc:start}--\ref{line:fundcyc:end}). 
If $\tilde C$ consists of fewer than $4\sqrt{m}$ natural edges, 
then $\tilde C$ is a short balanced simple-cycle separator. 
\label{step:fundcycle}
\item Otherwise, 
the fact that $\tilde C$ is long implies 
that it intersects components at many 
(more than $2\sqrt m$) 
consecutive levels of the component tree $\mathcal K$ 
(we say that $\tilde C$ intersects a component $K$ 
if $\tilde C$ has at least one edge in $K-X(K)$ and 
at least one in $G-K$).
The algorithm performs a binary search procedure on the range $[l,h)$, 
where $l,h$ are the minimum and maximum level $\dualnodelevel(v)$ of a vertex $v \in \tilde C$, respectively.
At each step of the binary search, the algorithm identifies 
\begin{itemize}
\item the median level $i_0 = \lfloor (l + h)/2 \rfloor$, 
\item the highest-level component $K_-$ intersected by $\tilde C$,
 whose level $i_-$ is smaller than $i_0$ and
 whose bounding cycle $X(K_-)$ has few (at most $\sqrt m$) natural edges, and 
\item the smallest-level component $K_+$ intersected by $\tilde C$,
 whose level $i_+$ is at least $i_0$ and
 whose bounding cycle $X(K_+)$ has few natural edges.
\end{itemize}
The monotonicity of the primal tree $T$ implies 
that the number of natural edges of $\tilde C$ between levels $i_-$ and $i_+$ 
(that is, the number of edges of $\tilde C$ in $K_- - K_+$) 
is at most $2\sqrt m$ 
(see Lemma~\ref{lem:CycleSeparator}).
\item If $W(G - K_-) > 3W/4$ 
we continue the binary search on the range $[l,i_-)$. 
Similarly, if $W(K_+) > 3W/4$ 
we continue the binary search on the range $[i_+,h)$.
\item Otherwise, 
the graph $G'$ induced by the edges of $X(K_-),X(K_+)$ 
and the edges of $\tilde C$ in $K_- - K_+$ 
is a biconnected planar graph with at most $4\sqrt m$ natural edges, 
none of whose faces weighs more than~$3W/4$. 
See Figure~\ref{fig:cyclesep}. 
The algorithm uses a simple greedy procedure, {\sc GreedyCycleSeparator}, 
to output a balanced simple cycle separator of $G'$ in $O(\sqrt m)$ time.
\end{enumerate}

\begin{algorithm}[h!]
\DontPrintSemicolon
 \SetAlgoLined
let $m$ be the number of natural edges in $G$ \label{line:fundcyc:m}\; 
let $W:=W(G)$ \label{line:fundcyc:W}\;
compute $e^*$ to be the $3/4$--balanced edge separator of $T^*$; \ \ \ let $e$ be the primal of $e^*$\;\label{line:fundcyc:start}
let $\tilde C$ denote the fundamental cycle defined by $e\not\in T$ and $T$\;\label{line:fundcyc:end}
\lIf{$\tilde C$ has at most $4 \sqrt{m}$ natural edges}{\Return $\tilde C$\ \ }
\tcc*{$\tilde C$ short enough}
let $l,h$ be the minimum and maximum level $\dualnodelevel(v)$ of a vertex $v \in \tilde C$ \label{line:lh}\;
\While(\tcc*[f]{binary search for $i_0$ in range $[l,h)$}){$l<h$}{
set all $i_0,i_-,i_+$ to $l+\floor{(h-l)/2}$\;
\Repeat(\tcc*[f]{sequential search for level $i_-<i_0$ with small boundary})
{$i_- < l$ or $X(K_-)$ has at most $\sqrt{m}$ natural edges\label{line:i-until}}
{$i_-:=i_--1$; 
let $K_-$ be the unique component at level $i_-$ that $\tilde C$ intersects \label{line:findK-}\;}
\lIf{$W(G\setminus K_-)>3W/4$}{$h:=i_-$; {\bf continue}} \label{line:k-}\;
\Repeat( \tcc*[f]{sequential search for level $i_+>i_0$ with small boundary})
{$i_+ > h$ or $X(K_+)$ has at most $\sqrt{m}$ natural edges\label{line:i+until}}
{$i_+:=i_++1$; 
let $K_+$ be the unique component of level $i_+$ that $\tilde C$ intersects \label{line:findK+}\;}
\lIf{$W(K_+)>3W/4$}{$l:=i_+$; {\bf continue}} \label{line:k+}\;
let $G'$ be the graph induced on $G$ by 
$X(K_-) \cup X(K_+) \cup\left( \tilde C \cap (K_- - K_+) \right)$ \label{line:G'}\;
\Return {\sc GreedyCycleSeparator}$(G')$\;\label{line:bothsmall}
} % end while
\caption{{\sc SimpleCycleSeparator}$(\mathcal G)$, $\mathcal{G}$ is the tuple $(G,\mathcal K,T,T^*)$\label{algo:cycle}}
\end{algorithm}

\noindent {\em Note:} To avoid clutter, 
the pseudocode of {\sc SimpleCycleSeparator} does not handle the boundary case 
where the loop in line~\ref{line:findK-} terminates without 
finding a level $K_-$ 
with $\abs{X(K_-)} < \sqrt m$. 
In this case 
there is no need to shortcut $\tilde C$ at a small level. 
Specifically, 
$\tilde C$ is entirely enclosed by $K_-$, 
the condition in line~\ref{line:k-} is considered false, 
and 
$X(K_-)$ is considered to be an empty set of edges. 
A similar statement applies to $K_+$.
\begin{lemma}\label{lem:K+}
The components $K_-$ and $K_+$ 
defined in Lines~\ref{line:findK-} and~\ref{line:findK+}, respectively, 
are well defined. 
\end{lemma}
\begin{proof}
Since $l \leq i_-\leq h$, the cycle $\tilde C$ must intersect some
component at level $i_-$. The monotonicity of the tree $T$
implies that once a rootward path leaves a component at some level it
never enters any component at that level again. The lemma follows
since $\tilde C$ is comprised of two rootward paths plus one edge.
\end{proof}

\begin{lemma}
For $l \leq i \leq h$, let $K_i$ be the unique component at level $i$
that is intersected by $\tilde C$.
If the weight not enclosed by $K_-$ is greater than $3W/4$ then there
exists a level $l \leq i_0 < i_- $ such that the weight
not enclosed by any $K_i$ with $l \leq i \leq i_0$ is at most
$3W/4$, and such that the weight
enclosed by any $K_i$ with $i_0 < i \leq h$ is at most
$3W/4$.
\end{lemma}
\begin{proof}
Let $K_l$ be the unique level-$l$ component that encloses
$\tilde C$. Since $\tilde C$ is a balanced separator, the weight
enclosed by $K_l$ is at least $W/4$. Hence the weight not enclosed
by $K_l$ is at most $3W/4$. Let $\tilde K$ be the component with
maximum level that is intersected by $\tilde C$ and whose enclosed weight is
at least $W/4$. Let $i_0$ be the level of $\tilde K$. By the above
argument $l \leq i_0$, and since the weight not enclosed by $K_-$
is greater than~$3W/4$, $i_0 < i_-$. By choice of $i_0$, any $K_i$
with $i_0 < i \leq h$ encloses at most $W/4$ weight.
\end{proof}
A symmetric lemma applies to the case where the weight enclosed by
$K_+$ is greater than $3W/4$. These lemmas show that, if $l$ or $h$
are updated (i.e., the binary search continues), there exists some
level $i_0$ in the new search range for which both conditions in Lines~\ref{line:k-}
and~\ref{line:k+} are false, and hence the binary search
procedure must eventually terminate.

\begin{lemma}\label{lem:CycleSeparator}
Let $G$ be a biconnected plane graph $G$ with $m$ natural edges and
face weights such that no face weighs more
than $3/4$ the total weight and such that no face consists of more than 3 edges.
The procedure {\sc SimpleCycleSeparator}
finds a $3/4$--balanced simple cycle separator in $G$ with at most
$4\sqrt m$ natural vertices.
\end{lemma}

\begin{proof}
Since $G$ has face size at most 3, the fundamental cycle $\tilde C$
is a $3/4$--balanced simple cycle
separator. If it consists of fewer than $4\sqrt m$ natural edges it is returned in Line~\ref{line:fundcyc:end}.
Otherwise, the lemma follows from the
correctness of {\sc GreedyCycleSeparator} (Lemma~\ref{lem:greedy-sep-correct}), provided that we show
that $G'$ is a biconnected subgraph of~$G$ with $4\sqrt m$ natural edges, none of
whose faces weighs more than $3W/4$. 
We assume that both $K_-$ and $K_+$ exist. 
The cases where one of them or both do not exist are similar.
Consider the cycles $X(K_-), X(K_+)$, and $\tilde C$. Since $\tilde C$
intersects both $K_-$ and $K_+$, $G'$ is biconnected. See
Figure~\ref{fig:cyclesep} for an illustration.

To establish the bound on the number of edges in $G'$, note that, by
choice of $X(K_-)$ and $X(K_+)$, they consist of fewer than $\sqrt m$
natural edges each. It remains to bound $\abs{\tilde C \cap (K_- - K_+)}$,
which consists of two paths in $T$ between $X(K_+)$ and $X(K_-)$.
We claim that $i_+ - i_- \leq \sqrt m$. To see this,
observe that by definition of $i_-$ and $i_+$, for every level $i_- <
i < i_+$, the bounding cycle of the unique level-$i$ component
intersected by~$\tilde C$ consists of more than $\sqrt m$ natural
edges. Since bounding cycles do not share natural edges
(Invariant~\ref{inv:disjoint-cycles}) and since there are
$m$ natural edges in $G$, it must be
that $i_+ - i_- + 1 \leq \sqrt m$.
Consider node $u$ on $X(K_+)$ and let
$u_0,u_1,\dots, u_k$ be the natural vertices on the $u$-to-root path in $T$ (i.e., $u_0 = u$).
Let $v$ be the first vertex that belongs to $X(K_-)$ on this path. 
Since $u$ is incident to a level $i_+$ face, $\dualnodelevel(u) \leq i_+$.
By invariant~\ref{inv:Tparent}, for any $i$, $\dualnodelevel(u_i) \leq
\dualnodelevel(u) - i = i_+ - i$. Hence, $\dualnodelevel(u_{\sqrt m}
\leq i_+ - \sqrt m \leq i_-$. This implies that $u_{\sqrt m}$ belongs
to a level cycle at level at most $i_-$. It follows that 
the number of natural vertices in the path in $T$ from $u$ to $v$ is at most $\sqrt m$.
We therefore conclude that $\abs{\tilde C \cap (K_- - K_+)}$ consists of at most $2\sqrt m$
natural vertices.

As for the face weights, the weight of a face $f'$ of $G'$ is the total
weight of the faces of $G$ that are enclosed in $G$ by the cycle
formed by the darts of $f'$. The faces of $G'$ that correspond to the
exterior of $X(K_-)$ and to the interior of $X(K_+)$ have weight at most
$3W/4$ by the conditions in Lines~\ref{line:k-} and~\ref{line:k+}. All
of the other faces of $G'$ are either enclosed (in $G$) by $\tilde C$
or not enclosed by $\tilde C$. Therefore, the weight of each of these
faces is at most $3W/4$ since $\tilde C$ is a balanced separator.
\end{proof}

\begin{figure}%[a]
\begin{center}
\includegraphics[width=0.4\textwidth]{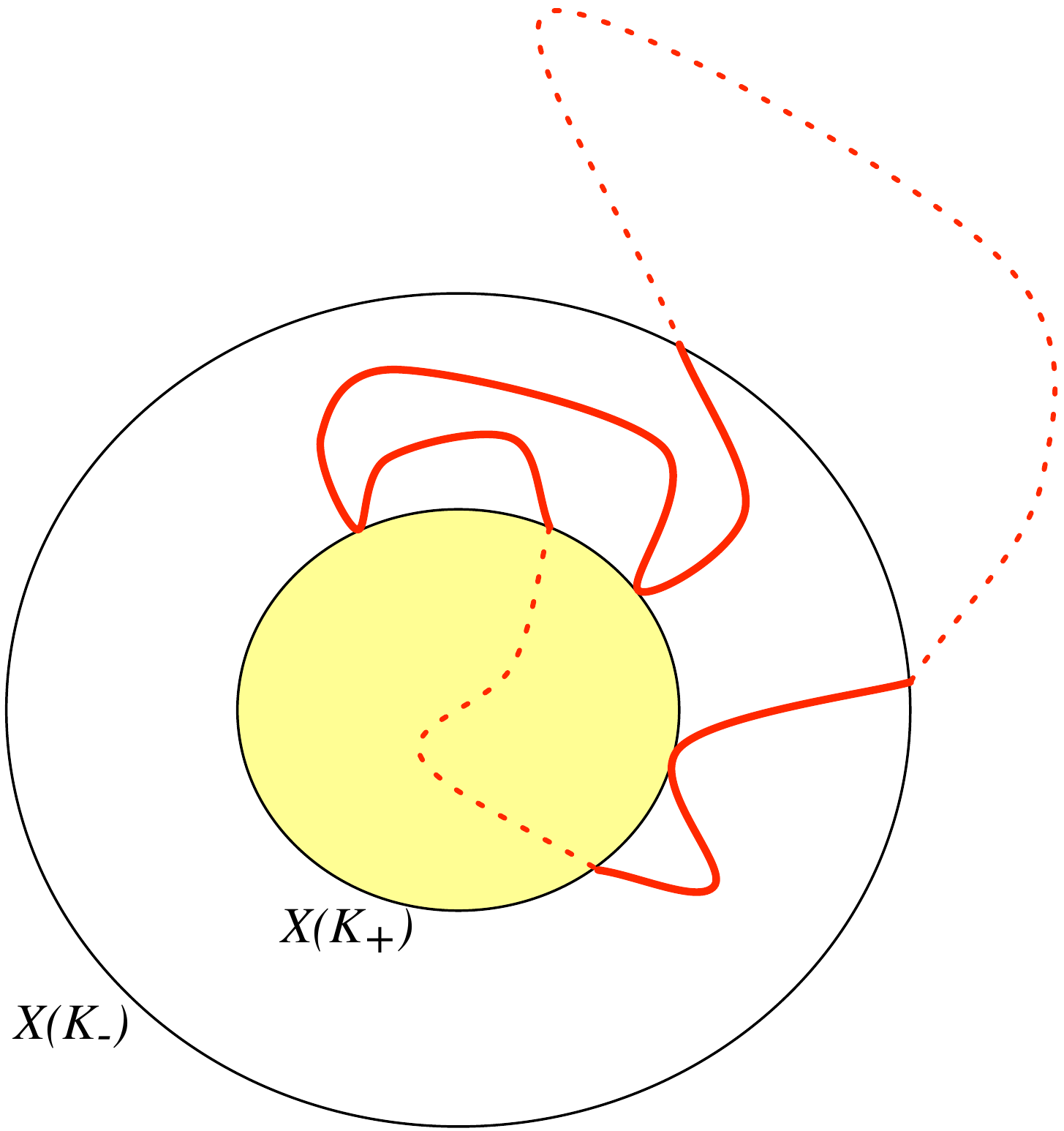}
\caption{Illustration of the cycle separator algorithm: the
 fundamental cycle $\tilde C$ is shown in red (solid and dashed). The level cycles
 $X(K_-)$ and $X(K_+)$ are indicated. The component $K_+$ is shaded
 yellow. $G'$ is the graph consisting of the edges of $\tilde C$ that
 belong to $K_- - K_+$ (solid red) as well as the edges of the cycles
 $X(K_-)$ and $X(K_+)$ (black). \label{fig:cyclesep}}
\end{center}
\end{figure}

Since any simple plane graph has at most three times as many edges as
vertices, and since triangulating and making biconnected any plane
graph $G$ requires $O(\abs{V(G)})$ edges, we can restate
Lemma~\ref{lem:CycleSeparator} in a more general form that does not
distinguish natural and artificial vertices.
\begin{corollary}\label{cor:CycleSeparator}
Let $G$ be a simple biconnected plane graph with face weights such that no face weighs more
than $3/4$ the total weight, and such that no face consists of more than 3 edges.
There exists a constant $c$ such that the procedure {\sc SimpleCycleSeparator}
finds a $3/4$--balanced simple cycle separator in $G$ whose length is
at most $c\sqrt{\abs{V(G)}}$.
\end{corollary}

We next discuss the procedure {\sc GreedyCycleSeparator} that finds a simple cycle in a
biconnected graph with no face weight exceeding $3W/4$.

\begin{lemma}\label{lem:greedy-sep-correct}
Let $G$ be a biconnected planar graph with $m$ edges and face weights
summing to $W$ such that no face weighs more than
$3W/4$. There exists an $O(m)$ algorithm that finds a balanced
simple-cycle separator in $G$.
\end{lemma}

\begin{proof}
The proof is constructive. Since $G$ is biconnected, every face
is a simple cycle. If $W(f) > W/4$ for any
face $f$, then $f$ is a balanced simple cycle
separator. Otherwise, use the following greedy algorithm. Initialize
the set $S$ with a single arbitrary face of $G$. Repeatedly add a
face $f$ to $S$, maintaining the invariant that the edges of $\delta_{G^*}(S)$ (the
boundary of $S$) form a simple
cycle, stopping when $W(S)$ first exceeds $W/4$. At this point, the
edges of $\delta_{G^*}(S)$ form a simple-cycle separator. It remains
to show that, in each iteration, there exists a face to add to $S$
while maintaining the invariant, and to show how such a face can be
quickly found.

We claim that, if the boundary of $S$ is a simple cycle, there exists
a face $f \notin S$ such that the boundary of $S \cup \{f\}$ is a
simple cycle. We prove the claim using an inductive argument
(illustrated in Figure~\ref{fig:greedy}).
Let $e$ be an edge of $\delta(S)$. Let $f'$ be the face not in $S$
to which $e$ belongs. Let $S'$ be $S \cup f'$. 
If $S'$ has a simple boundary, we are done. Otherwise, $\delta(S')$
consists of a set $\mathcal C$ of at least two simple cycles.
Consider the graph $G'$ obtained from $G$ by
deleting the non-boundary edges of $S'$. Each cycle in $\mathcal C$ consists of at least one edge of
$\delta(S)$ and one edge of $f'$, and bounds a 
connected subgraph of $G'$. Let $H$ be one such connected subgraph.
Since $G$ is biconnected, any two vertices in $H$ are
connected by two vertex-disjoint paths in $G$. These paths can
be transformed into vertex-disjoint paths in $H$ by rerouting along
the boundary of $H$. Such rerouting is possible since the boundary of $H$ is a simple cycle.
Hence $H$ is biconnected. Since $H$ contains
an edge of $S$, $H \cup S$ is biconnected as well.
Now repeat the argument with $G$ replaced by $S \cup H$. Since
$S \cup H$ is strictly smaller than $G$ (it does not include at least one edge
--- the one in $f' \cap C'$ for some cycle $C' \in \mathcal C$ that is
not the one bounding $H$), the inductive argument proves the claim.

We describe how, in each iteration, the greedy algorithm finds a face
$f$ such that the boundary of $S\cup \set{f}$ is a simple cycle. The algorithm maintains,
for each face $f$,
the number $\gamma(f)$ of consecutive maximal subpaths of
$\delta(S)$ on $f$. (A subpath may consist of a single vertex of
$\delta(S)$.) Note that $S \cup f$ has simple boundary iff
$\gamma(f) = 1$. The algorithm also maintains a list of the faces $f$ 
not enclosed by $S$ that have $\gamma(f)=1$ and at least one edge shared
with $S$. When a face $f$ is added
to the set $S$, the algorithm updates $\gamma(\cdot)$ as follows.
Let $S'$ denote $S \cup f$. Since $G$ is biconnected, $f$
is a simple cycle. Let $P$ be the subpath of $f$ that consists
of the boundary vertices and edges of $S'$ that do not belong to $S$. 
Let $f_1,f_2, \dots$ be the faces other than $f$ incident to edges or
vertices of $P$ in clockwise order along $f$. Note that the $f_i$'s are
not necessarily distinct (this is illustrated in
Figure~\ref{fig:greedy-alg}).
Let $P_i$ be the maximal subpath (possibly a single vertex)
of $P$ along which $f_i$ is incident to $P$. For each subpath $P_i$, let $e_i^-$ and $e_i^+$ be the edges 
preceding and following $P_i$ in $\delta(S')$, respectively. 
If both $e_i^-$ and $e_i^+$ belong to $f_i$, decrease $\gamma(f_i)$
by one. 
If neither $e_i^-$ nor $e_i^+$ belongs to $f_i$, increase $\gamma(f_i)$
by one. 
Otherwise, $\gamma(f_i)$ remains unchanged.
Since the work done when adding a face is proportional to the sum of degrees of the vertices of $P$, and since each edge or vertex
belongs to $P$ only once, when it first enters $S$, 
the total running time of this greedy
algorithm is $O(m)$.
\end{proof}

\begin{figure}[h!!]
\begin{center}
\includegraphics[width=0.4\textwidth]{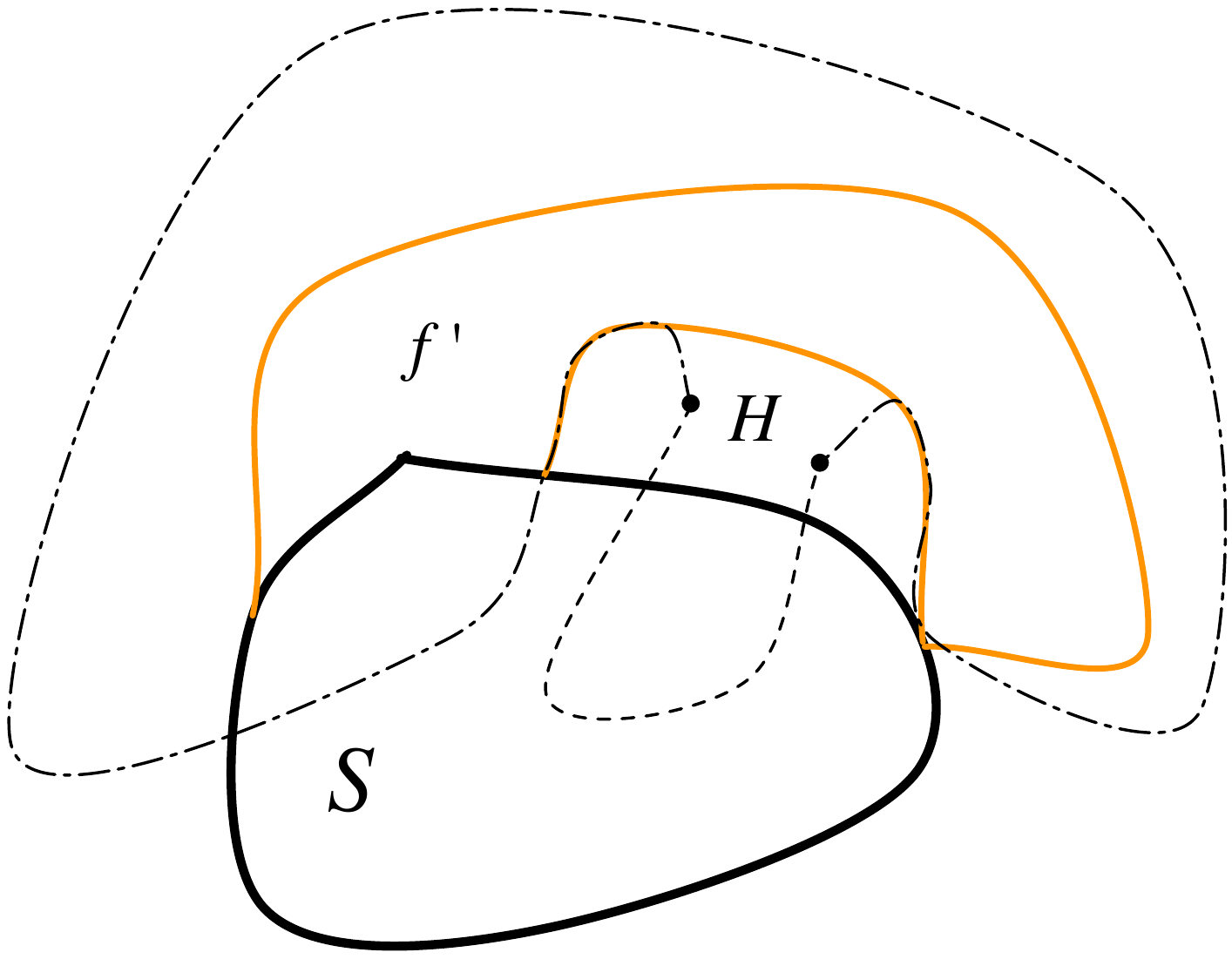}
\caption{Illustration of the proof of the auxiliary claim in
 Lemma~\ref{lem:greedy-sep-correct}. 
The boundary of the set $S$ is shown in solid thick black. The face $f'$ is
shown in solid orange. $H$ is a subgraph bounded by one of the simple
cycles that comprise the boundary of $S' = S \cup f'$. A pair of vertices in
$H$ is indicated, along with two vertex-disjoint paths between them
in the entire graph $G$. These paths can be transformed into disjoint
paths in $H$ by rerouting each of them along part of the simple cycle
bounding $H$. The subgraph $S \cup H$ is biconnected and strictly
smaller than $G$.
\label{fig:greedy}}
\end{center}
\end{figure}

\begin{figure}[h!!]
\begin{center}
\includegraphics[width=0.4\textwidth]{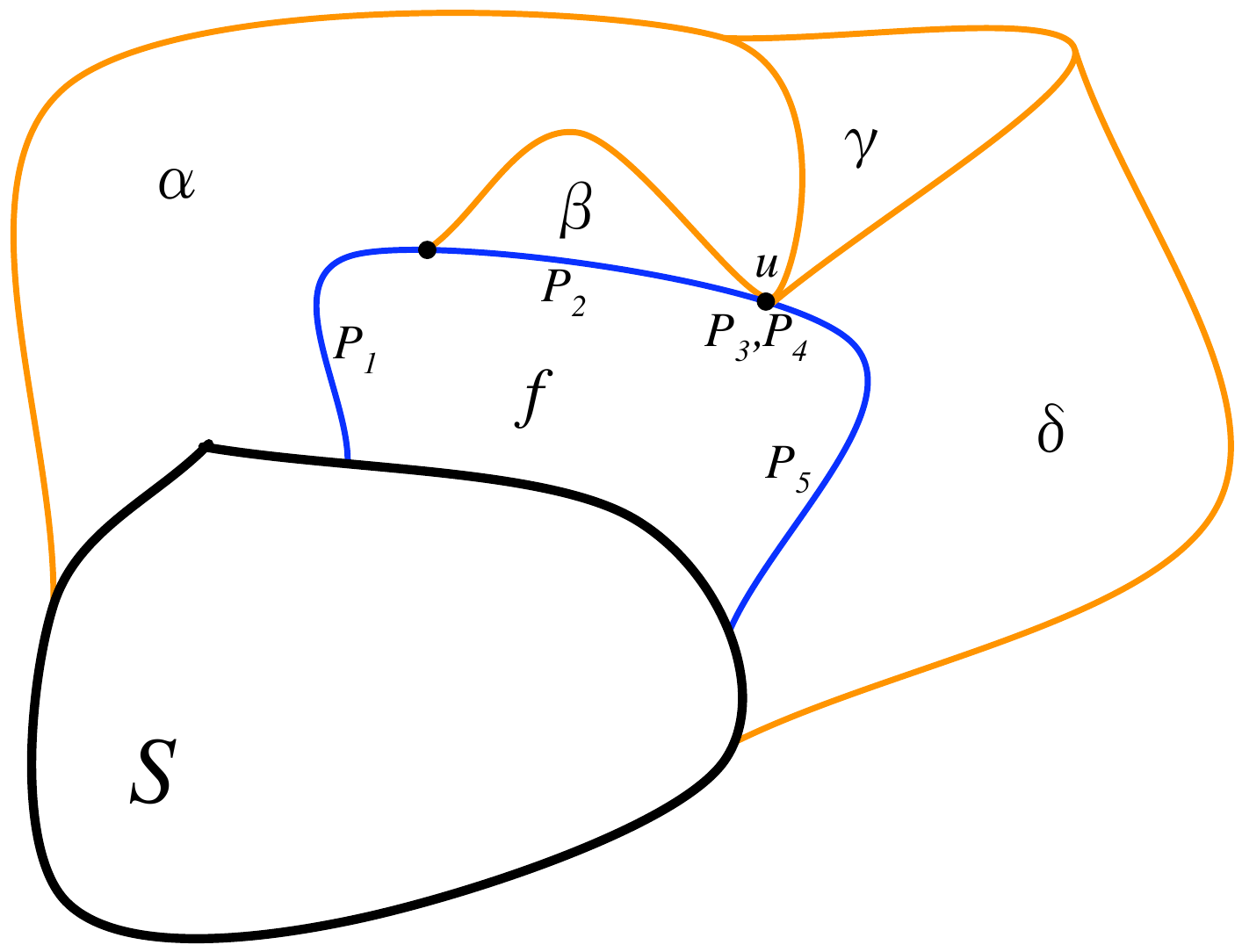}
\caption{Illustration of the algorithm 
in the proof of Lemma~\ref{lem:greedy-sep-correct}. 
The boundary of the set $S$ is shown in solid thick black. 
The path $P$ is solid blue. 
The list of faces $f_i$ is $f_1 = \alpha, f_2 = \beta,f_3 = \alpha, f_4 = \gamma, f_5 = \delta$. 
Note that the subpaths $P_3$ and $P_4$ both consist of the single vertex $u$.
Before $f$ is added to $S$, 
$\gamma(\alpha) = 1$. 
Since $e_1^-$, 
the edge of $\delta(S')$ preceding $P_1$, 
belongs to $\alpha$, 
and $e_1^+$, 
the edge of $\delta(S')$ following $P_1$, 
does not belong to~$\alpha$, 
$P_1$ does not contribute anything to~$\gamma(\alpha)$. 
Since both $e_3^-$ and $e_3^+$ do not belong to $\alpha$, 
$P_3$ contributes $+1$ to~$\gamma(\alpha)$. 
Hence, 
after $f$ is added to $S$, 
$\gamma(\alpha) = 2$.
\label{fig:greedy-alg}}
\end{center}
\end{figure}

\subsection{Efficient Implementation}
\label{sec:cyclesep-eff}

It is fairly straightforward to implement {\sc SimpleCycleSeparator} in linear time. In
this subsection we describe the data structures that support
implementing {\sc SimpleCycleSeparator} in sublinear time.
Our preprocessing procedure initializes the tree $T$ to be the 
tree $\bar T$ (Lemma~\ref{lem:Tparent}), the cotree $T^*$ to be the spanning tree of $\barG^*$, rooted at
$f_\infty$ whose edges are those not in $T$. $T$ is represented using
an Euler tour tree~\cite{HK99}. $T^*$ is represented
using a dynamic tree.
It also computes the component tree $\mathcal K$ and all level cycles
in $G$.
The component tree can be represented by a parent list. 
The level cycles are represented by splay trees~\cite{journals/jacm/SleatorT85}.

For every edge $e$ in $\barG$, let $P^{\mathcal K}(e)$ be the path in
$\mathcal K$ that consists of all components $K$ such that $e$ has one
endpoint in $X(K)$ and another endpoint strictly enclosed by
$X(K)$. We maintain these paths by storing the first and last components of $P^{\mathcal
 K}(e)$ in an array entry $\tt{R}(e)$. Note that $\tt{R}$ can be
generated in linear time in the preprocessing step as part of the
computation of the component tree.

To implement the algorithm efficiently we need to support the
following operations:
\begin{itemize}
\item Line~\ref{line:fundcyc:m} - Find the number of natural edges in
 $G$. This can be done by keeping track of the number $m_{\rm
 artificial}$ of artificial
 triangulation edges (Lemma~\ref{lem:boundary-bound} in
 Section~\ref{sec:maint} shows that this takes total $O(n)$ time over
 all calls to {\sc SimpleCycleSeparator}, so we may assume $m_{\rm
 artificial}$ is given). Let $F$ be the number of
 faces in $G$. $F$ can be found in $O(\log\abs{E(G)})$ amortized time by a
 single dynamic tree operation on~$T^*$. Since $G$ is triangulated, the number of natural edges
 is given by $3F/2 - m_{\rm artificial}$.
\item Line~\ref{line:fundcyc:W} - Find the total weight of faces in
 $G$. This is done using a single dynamic tree operation that returns
 the total weight of nodes (faces) in $T^*$
\item Line~\ref{line:fundcyc:start} - Find a balanced edge separator
 in $T^*$. This is done using a dynamic tree operation that returns
 a leafmost edge whose subtree has at least a certain total amount of
 weight ($W/4$ in our case). This operation takes $O(\log\abs{E(G)})$
 amortized time.
\item Line~\ref{line:fundcyc:end} - Find out if a path in $T$ consists
 of at most a given number $k$ of natural edges (in this case $k=4\sqrt{\abs{E(G)}}$). 
 This can be done by traversing the path, and takes $O(k \log\abs{E(G)})$ time (the
 number of artificial edges on a path is bounded by the number of
 holes of $G$, which, by Lemma~\ref{lem:num-holes}, is constant).
\item Line~\ref{line:lh} - Find least common ancestor in $T$. This can be done using
 the Euler tour representation in $O(\log\abs{E(G)})$
 amortized time. 
\item Line~\ref{line:findK-},~\ref{line:findK+} - Find leafmost
 ancestor in $T$ with level at most $i_-$. This can be done using
 the Euler tour representation in $O(\log\abs{E(G)})$
 amortized time. 
\item Find greatest-level component to which a particular edge
 belongs. This information is stored in the array $\tt{R}$.
\item Line~\ref{line:i-until},~\ref{line:i+until} - Query the number of
 edges in a level cycle. This takes $O(\log\abs{E(G)})$ time since
 level cycles are stored as splay trees.
\item Line~\ref{line:k-},~\ref{line:k+},~\ref{line:G'} - Query the
 weight enclosed by a short simple cycle $C$. This can be done by summing
 up the weight of subtrees as follows.
 For each cycle edge $e$, if $e \in T^*$ and
 the rootward endpoint of $e$ in $G^*$ is not enclosed by $C$, add
 the total weight of the subtree of $e$ in $T^*$ to the sum. If $e \in T^*$ and
 the rootward endpoint of $e$ in $G^*$ is enclosed by $C$, subtract
 the total weight of the subtree of $e$ in $T^*$ from the sum. Since
 querying the total weight of a subtree can be done using a dynamic
 tree operation in $O(\log\abs{E(G)})$ amortized time, and since the length
 of $C$ is $O(\sqrt{\abs{E(G)}})$, computing the total weight enclosed by $C$
 can be done in $O(\sqrt{\abs{E(G)}} \log \abs{E(G)})$ time.
\end{itemize}

This shows that, assuming that the auxiliary data structures are given
in the desired representation, {\sc SimpleCycleSeparator} can be
implemented in $O(\sqrt{\abs{E(G)}} \log^2{\abs{E(G)}})$ time.

\section{Maintaining the Representation of Regions Efficiently}
\label{sec:maint}

In this section we describe how to maintain the auxiliary data
structures that represent the regions and are required by {\sc SimpleCycleSeparator} 
throughout the recursive calls to {\sc RecursiveDivide}. 
These include maintaining the embedding of regions, the primal
spanning tree $T$ and its monotonicity invariant
(Invariant~\ref{inv:Tparent}), the cotree $T^*$, the component tree, and
the level cycles. 

Consider an iteration of {\sc RecursiveDivide} in which a
region $R$ is partitioned into two regions $R_0$ and $R_1$ along a cycle
separator $C$. 
Note that, due to the recursive nature of
the algorithm, we may assume that first the representations for $R_0$
are obtained and the algorithm is invoked recursively on $R_0$. After
the recursive call for $R_0$ 
is completed, all changes are undone, and the representations for
$R_1$ are obtained.
To achieve linear time for computing an $r$--division, it is required
that the update be done in sublinear time in the size of $R$. 
Essentially, the representation can be updated by performing only
local operations on the vertices and edges of the separator $C$. This
results in sublinear running time since
$C$ consists of $O(\sqrt{\abs{R}})$ edges, and since each operation in the
data structures we use takes $O(\log\abs{R})$ amortized time. 
Lemma~\ref{lem:recurrence} with $r=s$ shows that if the cycle can be 
computed in $O(\abs{R}^\beta)$ time (for some $\beta<1$), as 
established in Section~\ref{sec:cyclesep-eff}, 
and if the representation 
can be updated in $O(\abs{R}^\beta)$ time, 
then the total running time of the algorithm is linear. This
establishes the running time claimed in Theorem~\ref{thm:decomposition}.

\subsection{Maintaining the Embedding}
The embedding is represented by using a splay
tree~\cite{journals/jacm/SleatorT85} for each vertex $v$ of $R$ to store the
edges incident to $v$ in cyclic order.
Conceptually, $R_0$ is obtained from $R$ by merging all the faces
enclosed by $C$ into a single face.
Merging the faces enclosed by $C$ corresponds to deleting all
edges of $R$ strictly enclosed by $C$. 
Similarly, $R_1$ is obtained
from $R$ by deleting all edges not enclosed by $R$. We do not perform
these deletions explicitly since that would not be efficient.

We describe the procedure for $R_0$. The procedure for $R_1$ is
similar. We start by ``cutting off'' the edges that do not belong to
$R_0$. To do this we reconstruct the cyclic order of edges around the
vertices of the cycle~$C$. Consider a vertex $v$ of $C$. Let $\pi_v$ be the
cyclic permutation of edges around $v$ in $R$. There are
$2$ edges of $C$ that are incident to $v$. These cycle edges
partition $\pi_v$ into $2$ intervals. The edges in one interval
belong to $R_0$ and those in the other interval do not. Using $O(1)$
cut operations, each taking $O(\log\abs{R})$ amortized time, we cut off the
edges that do not belong to $R_0$. Since the length of $C$ is $O(\sqrt
{\abs{R}})$ (it consists of $O(\sqrt{\abs R})$ natural nodes, and at
most 12 artificial nodes), the total time to update all vertices
of $C$ is $O(\sqrt{\abs{R}} \log{\abs{R}})$. Note that there is no need to
explicitly delete the vertices enclosed by $C$ from the representation
because they are no longer pointed to from any vertex in $R_0$.

Observe that the definition of $R_0$ in Line~\ref{line:R_i} of {\sc
 RecursiveDivide} and the definition of $R'_0$ in
Line~\ref{line:triangulation} of the recursive call of {\sc
 RecursiveDivide} on $R_0$ are equivalent to deleting all the artificial triangulation
edges incident to $C$, and then triangulating the new artificial face
of $R_0$ with a new artificial node and new
artificial edges. This can be done using a constant number of link and
cut operations per node of $C$, each requiring $O(\log\abs{R})$ amortized
time. The number of triangulation edges is bounded by the number of
boundary vertices of $R_0$. Note, however, that even though our algorithm
separates according to the number of boundary vertices every third
recursive call, there is no guarantee on the ratio between the number
of vertices in $R_0$ and the number of boundary vertices of $R_0$. Instead,
we bound the total time required for creating all of the triangulation edges throughout
the entire execution of the recursive algorithm. Recall that for a
node $x$ of $\mathcal T$, $b(x)$ denotes the number of boundary nodes
of $R_x$.

\begin{lemma}\label{lem:boundary-bound}
$\sum_{x \in \mathcal T} b(x)\log\abs{R_x} = O(\abs{R_{\hat x}})$, 
where $\hat x$ is the root of $\mathcal T$ (i.e., $R_{\hat x}$ is the
input graph $\barG$).
\end{lemma}
\begin{proof}
Consider a node $x$ of $\mathcal T$. Let $C_x$ denote the
cycle used to separate $R_x$.
Observe that a vertex $v$ on $C_x$ appears as a boundary vertex in the
regions of exactly two descendants of $x$ in each level $\ell >
\ell(x)$ of $\mathcal T$, {\em unless} $v$ appears on another cycle
separator~$C_y$ for some descendant $y$ of $x$.
Since the depth of the subtree of $\mathcal T$ rooted at $x$ is
bounded by $c_1 \log \abs{R_x}$ for some constant $c_1$ (the size of regions decreases by a constant factor
every 3 levels) we have:
\begin{equation}
\sum_{x \in \mathcal T} b(x) \leq \sum_{x \in \mathcal T}
2c_1 \abs{C_x} \log \abs{R_x}
\end{equation}
Note that if $v$ does appear on $C_y$ for some descendant $y$ of $x$, 
its additional occurrences as a boundary vertex are accounted for by
the contribution of $C_y$ to the sum.

Let $T(n)$ be the maximum of $\sum_{x \in \mathcal T} 2c_1\abs{C_x} \log\abs{R_x}$ over all
$n$-vertex triangulated biconnected planar input graphs. 
Let $k$, $y_1,\dots y_k$, $\alpha_1, \dots \alpha_8$ be as in
Lemma~\ref{lem:phase1}.
Let $\mathcal T_{x,y_1,\dots,y_k}$ denote the subtree of $\mathcal T$
induced by nodes that are descendants of $x$ but not descendants of
any $y_i$. Note that the number of internal nodes in $\mathcal
T_{x,y_1,\dots,y_k}$ is $k-1$, which is at most 7, and that for any internal node $z$ of
$\mathcal T_{x,y_1,\dots,y_k}$, $\abs{R_z}\leq \abs{R_x}$, and $\abs{C_z} \leq \abs{C_x} \leq
c\sqrt{\abs{R_x}}$. Hence, 
\begin{equation}
T(n) \leq 2c_1 \cdot 7c \sqrt n \log^2 n + \max_{\{\alpha_i\}} \sum
T(\alpha_i n)
\end{equation}
Lemma~\ref{lem:recurrence} shows that $T(n)$ is $O(n)$. 
\end{proof}

\subsection{Representation of the Spanning Tree $T$ and its Cotree $T^*$}
Let $T$ and $T^*$ denote the primal and dual spanning trees of $R$,
respectively, at the time the
cycle separator $C$ of $R$ is found. 
In what follows we describe how to convert $T$ into the spanning tree
$T_0$ of~$R_0$. The procedure for converting $T$ into the spanning
tree $T_1$ of $R_1$ is similar.
 
Let $h$ be the new artificial face of $R_0$.
The tree $T$ is cut at each vertex $v$ of $h$
so that the only edges adjacent to $v$ in $T$ are those that belong to
$R_0$ (i.e., not strictly enclosed by $h$). 
This breaks $T$ into a forest~$F$ in $O(\sqrt{\abs{R}} \log
\abs{R})$ time. Let $F_0$ be the set of trees in $F$ that consist of
edges of $R_0$.
The cotree $T^*$ is cut at the edges of
$h$. This disconnects $T^*$ into a forest $F^*$ in $O(\sqrt{\abs{R}} \log
\abs{R})$ time. Let $F^*_0$ be the set of trees in $F^*$ that consist
of duals of edges of $R_0$.

The algorithm connects the forest $F_0$ into a spanning tree $T_0$ of $R_0$
while maintaining the monotonicity invariant Invariant~\ref{inv:Tparent}).

\begin{lemma} \label{lem:no-parent-not-min}
Let $D$ be a cycle of natural vertices in $R$. Let $v$ be a vertex of $D$ such that the
parent of $v$ in $T$ is enclosed by $D$, where enclosure is defined
with respect to the root of $T$. There exists a vertex $u
\in D$ with $\dualnodelevel(u) < \dualnodelevel(v)$. 
\end{lemma}
\begin{proof}
Since the root of $T$ is not enclosed by $D$, there
must be a vertex $u\in D$ on the path in $T$ from $v$ to the root of
$T$. By Invariant~\ref{inv:Tparent}, $\dualnodelevel(u) < \dualnodelevel(v)$.
\end{proof}

Recall that, after deleting all the triangulation edges incident to
$h$, $h$ is retriangulated with a new
artificial vertex $x_h$ and new artificial edges $vx_h$ for each
vertex of $h$. 
Let $v_h$ be a vertex which minimizes $\dualnodelevel(v)$ among all
vertices of $h$. The algorithm attaches $x_h$ as a child of $v_h$ in $T$.
It then considers every vertex $v$ of $h$.
If $v$ does not have a parent in $F_0$ (that is, the parent of $v$ in
$T$ was enclosed by $h$ in $R$, so after cutting $T$, $v$ is a root of a
tree in the forest $F_0$) then the algorithm sets $x_h$
to be the parent of $v$ in $F_0$ by adding the
artificial triangulation edge $vx_h$ to $F_0$.
This operation is performed by linking the tree of $F_0$ whose root is $v$ to $x_h$.
If $v$ does have a parent in $F_0$ we add $vx_h$ to $F^*_0$.
At the end of the process, each edge of $R_0$
is either in $F_0$ or in $F^*_0$, and $F_0$ contains no cycles. Hence
$F_0$ and $F^*_0$ are spanning trees of $R_0$ and its dual,
respectively.

\begin{lemma}
After the update, the tree $F_0$ satisfies Invariant~\ref{inv:Tparent}.
\end{lemma}
\begin{proof}
The only natural nodes whose parent is set by the update procedure are
those on the new hole $h$ of $R_0$ that do not have a parent after $F$
was cut at the vertices of $h$ . By Lemma~\ref{lem:no-parent-not-min},
for each such vertex $v$, 
$\dualnodelevel(v_h) < \dualnodelevel(v)$. Since the algorithm sets
$x_h$ as the parent of $v$, and $v_h$ as the parent of $x_h$, any
natural vertex of $R_0$ that is a proper ancestor of $v$ is an
ancestor of $v_h$. Invariant~\ref{inv:Tparent} is maintained since the ancestors of $v_h$ are not changed by the
update procedure.

The number of operations required to obtain the updated representation
of the tree and cotree for $R_0$ is proportional
to the number of nodes of
$h$. By Lemma~\ref{lem:boundary-bound}, the total update cost over
all recursive calls to {\sc RecursiveDivide} is $O(\abs{\bar G})$.
Obtaining the representation of the tree and cotree for $R_1$ is
similar.
\end{proof}

\subsection{Representation of Component Boundaries}
The component tree $\mathcal K$ is computed at the preprocessing
step and represented by parent pointers.
We do not update the component tree when {\sc RecursiveDivide}
partitions the graph. We only maintain and update the level-cycles
(the boundaries of components). Note that at any given recursive call
$\mathcal K$ correctly represents the enclosure relations between
components in the current subgraph. However, $\mathcal K$ also
contains components that no longer belong to the current
subgraph. This does not pose a problem since we only use $\mathcal K$
for parent queries.

For every edge $e$ in $G$, let $P^{\mathcal K}(e)$ be the path in
$\mathcal K$ that consists of all components $K$ such that $e$ has one
endpoint in $X(K)$ and another endpoint strictly enclosed by
$X(K)$. The algorithm maintains these paths by storing the first and last components of $P^{\mathcal
 K}(e)$ in an array entry $\tt{R}(e)$. Note that $\tt{R}$ is
generated in linear time at the preprocessing step and is used when
working with subgraphs at different recursive levels. $\tt{R}(e)$ need
not be updated for natural edges. However, whenever a new triangulation
edge $e$ is introduced, the algorithm must store the first and last
components of $P^{\mathcal K}(e)$ in $\tt{R}(e)$. The way this is done
is described at the end of this subsection.

We now describe how we update the representation of the level
cycles. By construction, the cycle separator $C$ intersects at most
$\sqrt m$ level-cycles (at most a single cycle at each level between $i_-$ and
$i_+$, see Lemma~\ref{lem:K+}) 
These are the only level-cycles that need to be updated. 
We describe the procedure for $R_0$. The one for $R_1$ is similar. 
As we mentioned above, these updates are first done for $R_0$. When the
recursive call for $R_0$ is completed, all changes are undone and we
perform the updates for $R_1$.

We first identify the components that are
bisected by $C$. 
This is done by traversing the cycle $C$. For each edge $e$ 
of $C$ we mark each component $K$ in
$P^{\mathcal K}(e)$ whose level is at least $i_-$ by the endpoint of $e$ that belongs to
$X(K)$. This can be done by starting at the last component of
$P^{\mathcal K}(e)$ (Since $C$ does not
contain edges of any component with level greater than $i_+$, the
level of the last component of $P^{\mathcal K}(e)$ is at most $i_+$), and
following parent pointers in the component tree $\mathcal K$ until either the other end
of $P^{\mathcal K}(e)$ is encountered or until a component of level
$i_-$ is encountered.

After traversing $C$, we have identified, for each marked component
$K$, two vertices $u$ and $v$ on $X(K)$. We cut the cycle $X(K)$ at $u$
and at $v$. Since cycle levels are represented using splay trees,
this takes $O(\log\abs{R})$ amortized time. This results in two paths,
let $X_0(K)$ be the one that 
consists of edges of $R_0$. The endpoints $u$ and $v$ of $X_0(K)$ are incident to
the new hole $h$ of $R_0$. 
We add a dummy edge $uv$ to turn $X_0(K)$ back into a cycle.
Note that dummy edge $uv$ represents the two triangulation edges
$vx_h$ and $x_hu$ that connect $u$ and $v$ in $R_0$. We do not
explicitly add those artificial edges 
to form a cycle since artificial edges (as well as the artificial node
$x_h$) might be replaced with new ones
at subsequent recursive calls. We cannot afford to update such edges
when retriangulating since each triangulation edge might appear in many
level cycles. To be able to associate a dummy edge with the
corresponding triangulation edges,
the algorithm maintains, for every natural vertex $v$, a list of the
artificial triangulation edges incident to $v$. At any given time
there are at most 12 items in the list. Note that, by
Lemma~\ref{lem:boundary-bound}, we can afford to update these lists
when retriangulating holes.

Note that, even though level cycles are no longer disjoint,
Invariant~\ref{inv:disjoint-cycles} is maintained. The reason is that
level cycles may share
artificial triangulation edges, but no natural edges.

We now describe how to update $\tt{R}$ with new triangulation edges.
For each triangulation edge that belongs to some level cycle, the
algorithm stores pointers to the lowest and deepest corresponding
components in $\mathcal K$. These can be identified at the time the
corresponding dummy edge is inserted into the representation of
level-cycle. When retriangulating a hole, every new triangulation edge
inherit this information from the triangulation edge it replaced.
Recall that all triangulation edges of a specific hole $h$ are incident to
a single artificial vertex $x_h$. Going over the edges incident to
$x_h$ in cyclic order, and knowing the level cycles to which each edge
belongs, it is easy to infer the corresponding paths $P^{\mathcal K}(e)$,
and store their endpoints in the array~$\tt{R}$.

\subsection{Representation of Vertex Weights}
Our algorithm separates regions according to three possible balance
criteria: number of natural vertices, number of boundary vertices, and
number of holes. In this subsection we discuss how these quantities
are represented and maintained. 

As we mention in Section~\ref{sec:num-of-vertices} and in the
beginning of Section~\ref{sec:cyclesep}, vertex weights are handled by
transferring the weight to some adjacent face. A cycle that is
balanced with respect to the face weights is balanced with respect to
the original vertex weights as well.

\newcommand{\naturalw}{{\tt natural\_weight}}
\newcommand{\boundary}{{\tt boundary\_weight}}
\newcommand{\hole}{{\tt hole\_weight}}
We define three types of face weights: $\naturalw$, $\boundary$ and $\hole$ such
that balancing according to each weight type is equivalent to
balancing according to the corresponding quantity.
Initially, all face weights are initialized to zero.
For each vertex $v$ of $\barG$, the algorithm associate
with $v$ an
arbitrary adjacent face $f$, and increases $\naturalw(f)$ by one. 
It maintains this vertex-face association in an array ${\tt F}$ so that, for every vertex,
the associated face can be queried in constant time.
Since the input graph $\barG$ is triangulated, each face has weight at
most 3, so no face carries more than $3/4$ of the total $\naturalw$
(recall that the cycle separator is not invoked on graphs with fewer
than $s$ vertices).

We now describe how to update the weights when separating a region $R$
along a cycle separator $C$ into subregions $R_0$ and $R_1$. We
describe the procedure for $R_0$. The one for $R_1$ is similar.
The algorithm checks, for each (natural) boundary vertex $v$ of $R_0$
if ${\tt F}[v]$ is a face of $R_0$. This can be done in $O(\log
\abs{R_0})$ time using the cotree $T^*$ for $R_0$. If ${\tt F}[v]$ is
not a face of $R_0$, the algorithm chooses an arbitrary face $f$ of $R_0$
that is adjacent to $v$, sets ${\tt F}[v]$ to $f$, and increases
$\naturalw(f)$ by one. This guarantees that each vertex of $R_0$
assigns unit weight to one associated face of $R_0$. Note also that
for each face $f$ of $R_0$, all vertices assigning weight to $f$ are
vertices of $R_0$. Hence, balancing $R_0$ by $\naturalw$ is equivalent to
balancing by the number of natural vertices of $R_0$. The time to
perform the update is $O(\abs{C} \log \abs{R})$, which, by
Lemma~\ref{lem:boundary-bound}, is $O(\abs{\barG})$ over the entire
execution of the algorithm.

To update $\boundary$, the algorithm sets $\boundary(f)$ to one for
every new triangulation face created when triangulating
$R_0$. Similarly, to update $\hole$, the algorithm sets to one
$\hole(f)$ for one arbitrary new triangulation face of the new hole
formed when cutting $R$ along $C$. Again, by
Lemma~\ref{lem:boundary-bound}, this takes total $O(\abs{\barG})$ time.

\newpage
\bibliographystyle{alpha}
\bibliography{planarseparators}

\newcommand{\etalchar}[1]{$^{#1}$}
\begin{thebibliography}{{\L}NSWN12}

\bibitem[ABH{\etalchar{+}}04]{ABHVW04}
Umut~A. Acar, Guy~E. Blelloch, Robert Harper, Jorge~L. Vittes, and Shan
  Leung~Maverick Woo.
\newblock Dynamizing static algorithms, with applications to dynamic trees and
  history independence.
\newblock In {\em Proceedings of the 15th Annual ACM-SIAM Symposium on Discrete
  Algorithms (SODA)}, pages 531--540, 2004.

\bibitem[ABT04]{journals/jal/ArgeBT04}
Lars Arge, Gerth~St{\o}lting Brodal, and Laura Toma.
\newblock On external- memory {MST}, {SSSP} and multi-way planar graph
  separation.
\newblock {\em J. Alg.}, 53(2):186--206, 2004.

\bibitem[ADGM06]{journals/jea/AleksandrovDGM06}
Lyudmil Aleksandrov, Hristo~Nicolov Djidjev, Hua Guo, and Anil Maheshwari.
\newblock Partitioning planar graphs with costs and weights.
\newblock {\em ACM J.\ Exp.\ Alg.}, 11, 2006.

\bibitem[AHdLT05]{AHLT05}
Stephen Alstrup, Jacob Holm, Kristian de~Lichtenberg, and Mikkel Thorup.
\newblock Maintaining information in fully dynamic trees with top trees.
\newblock {\em TALG}, 1(2):243--264, 2005.

\bibitem[And86]{journals/jct/Andreae86}
Thomas Andreae.
\newblock On a pursuit game played on graphs for which a minor is excluded.
\newblock {\em J.\ Comb.\ Th.\ B}, 41(1):37--47, 1986.

\bibitem[AP89]{Arn}
Stefan Arnborg and Andrzej Proskurowski.
\newblock Linear time algorithms for {NP}-hard problems restricted to partial
  $k$-trees.
\newblock {\em Discrete Appl.\ Math.}, 23(1):11--24, 1989.

\bibitem[AST90]{AlonST90}
Noga Alon, Paul~D. Seymour, and Robin Thomas.
\newblock A separator theorem for nonplanar graphs.
\newblock {\em JAMS}, 3(4):801--808, 1990.

\bibitem[AvWZ13]{AWZ13}
Lars Arge, Freek van Walderveen, and Norbert Zeh.
\newblock Multiway simple cycle separators and {I}/{O}-efficient algorithms for
  planar graphs.
\newblock In {\em Proceedings of the 24th Annual ACM-SIAM Symposium on Discrete
  Algorithms (SODA)}, pages 901--918, 2013.

\bibitem[BLR10]{journals/jacm/BiswalLR10}
Punyashloka Biswal, James~R. Lee, and Satish Rao.
\newblock Eigenvalue bounds, spectral partitioning, and metrical deformations
  via flows.
\newblock {\em JACM}, 57(3), 2010.

\bibitem[Cab12]{journals/algorithmica/Cabello12}
Sergio Cabello.
\newblock Many distances in planar graphs.
\newblock {\em Algorithmica}, 62(1--2):361--381, 2012.

\bibitem[Chu90]{ChungSurvey}
Fan R.~K. Chung.
\newblock Separator theorems and their applications.
\newblock In {\em Paths, Flows, and VLSI-Layout, Algorithms and Combinatorics},
  pages 17--34. Springer-Verlag, 1990.

\bibitem[CR10]{CabelloR10}
Sergio Cabello and G\"{u}nter Rote.
\newblock Obnoxious centers in graphs.
\newblock {\em SIAM J.\ Discrete Math.}, 24(4):1713--1730, 2010.

\bibitem[Dji82]{Djidjev1982}
Hristo~Nicolov Djidjev.
\newblock On the problem of partitioning planar graphs.
\newblock {\em SIAM J.\ Algeb.\ Disc. Meth.}, 3:229--240, 1982.

\bibitem[Dji85]{DjidjevSerdica1985LinearTime}
Hristo~Nicolov Djidjev.
\newblock A linear algorithm for partitioning graphs of fixed genus.
\newblock {\em Serdica. Bulgariacae mathematicae publicationes},
  11(4):369--387, 1985.

\bibitem[DV97]{DjidjevVenkatesan97}
Hristo~Nicolov Djidjev and Shankar~M. Venkatesan.
\newblock Reduced constants for simple cycle graph separation.
\newblock {\em Acta Informatica}, 34:231--243, 1997.

\bibitem[EFN12]{conf/soda/EricksonFN12}
Jeff Erickson, Kyle Fox, and Amir Nayyeri.
\newblock Global minimum cuts in surface embedded graphs.
\newblock In {\em Proceedings of the 23rd Annual ACM-SIAM Symposium on Discrete
  Algorithms (SODA)}, pages 1309--1318, 2012.

\bibitem[EIT{\etalchar{+}}92]{journals/jal/EppsteinITTWY92}
David Eppstein, Giuseppe~F. Italiano, Roberto Tamassia, Robert~Endre Tarjan,
  Jeffery Westbrook, and Moti Yung.
\newblock Maintenance of a minimum spanning forest in a dynamic plane graph.
\newblock {\em J. Alg.}, 13(1):33--54, 1992.

\bibitem[Epp03]{conf/soda/Eppstein03}
David Eppstein.
\newblock Dynamic generators of topologically embedded graphs.
\newblock In {\em Proceedings of the 14th Annual ACM-SIAM Symposium on Discrete
  Algorithms (SODA)}, pages 599--608, 2003.

\bibitem[FEMPS13]{FMPS13}
Eli Fox-Epstein, Shay Mozes, Phitchaya~Mangpo Phothilimthana, and Christian
  Sommer.
\newblock Short and simple cycle separators in planar graphs.
\newblock In {\em Proceedings of the 14th Meeting on Algorithm Engineering and
  Experiments (ALENEX)}, pages 26--40, 2013.

\bibitem[FR06]{journals/jcss/FakcharoenpholR06}
Jittat Fakcharoenphol and Satish Rao.
\newblock Planar graphs, negative weight edges, shortest paths, and near linear
  time.
\newblock {\em J.\ Comput.\ Syst. Sci.}, 72(5):868--889, 2006.

\bibitem[Fre87]{journals/siamcomp/Frederickson87}
Greg~N. Frederickson.
\newblock Fast algorithms for shortest paths in planar graphs, with
  applications.
\newblock {\em SIAM J.\ Comput.}, 16(6):1004--1022, 1987.

\bibitem[Fre97]{Frederickson97}
Greg~N. Frederickson.
\newblock A data structure for dynamically maintaining rooted trees.
\newblock {\em J. Alg.}, 24:37--65, 1997.

\bibitem[GHT84]{journals/jal/GilbertHT84}
John~R. Gilbert, Joan~P. Hutchinson, and Robert~Endre Tarjan.
\newblock A separator theorem for graphs of bounded genus.
\newblock {\em J. Alg.}, 5(3):391--407, 1984.

\bibitem[GM90]{conf/sigal/GazitM90}
Hillel Gazit and Gary~L. Miller.
\newblock Planar separators and the {Euclidean} norm.
\newblock In {\em SIGAL}, pages 338--347, 1990.

\bibitem[Goo95]{journals/jcss/Goodrich95}
Michael~T. Goodrich.
\newblock Planar separators and parallel polygon triangulation.
\newblock {\em J.\ Comput.\ Syst. Sci.}, 51(3):374--389, 1995.

\bibitem[Hal76]{halin}
Rudolf Halin.
\newblock {$S$}-functions for graphs.
\newblock {\em J.\ Geom.}, 8(1-2):171--186, 1976.

\bibitem[HK99]{HK99}
Monika~Rauch Henzinger and Valeria King.
\newblock Randomized fully dynamic graph algorithms with polylogarithmic time
  per operation.
\newblock {\em JACM}, 46(4):502--516, 1999.

\bibitem[HKRS97]{journals/jcss/HenzingerKRS97}
Monika~Rauch Henzinger, Philip~Nathan Klein, Satish Rao, and Sairam
  Subramanian.
\newblock Faster shortest-path algorithms for planar graphs.
\newblock {\em J.\ Comput.\ Syst. Sci.}, 55(1):3--23, 1997.

\bibitem[HSW{\etalchar{+}}09]{journals/jea/HolzerSWPZ09}
Martin Holzer, Frank Schulz, Dorothea Wagner, Grigorios Prasinos, and
  Christos~D. Zaroliagis.
\newblock Engineering planar separator algorithms.
\newblock {\em ACM J.\ Exp.\ Alg.}, 14, 2009.

\bibitem[INSWN11]{conf/stoc/ItalianoNSW11}
Giuseppe~F. Italiano, Yahav Nussbaum, Piotr Sankowski, and Christian
  Wulff-Nilsen.
\newblock Improved algorithms for min cut and max flow in undirected planar
  graphs.
\newblock In {\em Proceedings of the 43rd Annual ACM Symposium on Theory of
  Computing (STOC)}, pages 313--322, 2011.

\bibitem[JV82]{JV82}
Donald~B. Johnson and Shankar~M. Venkatesan.
\newblock Using divide and conquer to find flows in directed planar networks in
  ${O}(n^{3/2} \log n)$ time.
\newblock In {\em 20th Allerton Conference on Communication, Control, and
  Computing}, pages 898--905, 1982.

\bibitem[Kel06]{journals/siamcomp/Kelner06}
Jonathan~A. Kelner.
\newblock Spectral partitioning, eigenvalue bounds, and circle packings for
  graphs of bounded genus.
\newblock {\em SIAM J.\ Comput.}, 35(4):882--902, 2006.

\bibitem[KKS11]{KKS}
{Ken-ichi} Kawarabayashi, Philip~Nathan Klein, and Christian Sommer.
\newblock Linear-space approximate distance oracles for planar, bounded-genus,
  and minor-free graphs.
\newblock In {\em Proceedings of the 38th International Colloquium on Automata,
  Languages and Programming (ICALP)}, pages 135--146, 2011.

\bibitem[KMS13]{KMS13}
Philip~Nathan Klein, Shay Mozes, and Christian Sommer.
\newblock Structured recursive separator decompositions for planar graphs in
  linear time.
\newblock In {\em Proceedings of the 45th Annual ACM Symposium on Theory of
  Computing (STOC)}, 2013.
\newblock To appear.

\bibitem[KR10]{conf/focs/KawarabayashiR10}
{Ken-ichi} Kawarabayashi and Bruce~A. Reed.
\newblock A separator theorem in minor-closed classes.
\newblock In {\em Proceedings of the 51st Annual Symposium on Foundations of
  Computer Science (FOCS)}, pages 153--162, 2010.

\bibitem[KS98]{journals/algorithmica/KleinS98}
Philip~Nathan Klein and Sairam Subramanian.
\newblock A fully dynamic approximation scheme for shortest paths in planar
  graphs.
\newblock {\em Algorithmica}, 22(3):235--249, 1998.

\bibitem[{\L}NSWN12]{LNSWN12}
Jacub {\L}{\k{a}}cki, Yahav Nussbaum, Piotr Sankowski, and Christian
  Wulff-Nilsen.
\newblock Single source - all sinks max flows in planar digraphs.
\newblock In {\em Proceedings of the 53rd Annual Symposium on Foundations of
  Computer Science (FOCS)}, pages 599--608, 2012.

\bibitem[{\L}S11]{conf/esa/LackiS11}
Jakub {\L}{\k{a}}cki and Piotr Sankowski.
\newblock Min-cuts and shortest cycles in planar graphs in {$O(n \log\log n)$}
  time.
\newblock In {\em Proceedings of the 19th European Symposium on Algorithms
  (ESA)}, pages 155--166, 2011.

\bibitem[LT79]{journals/siamjam/LiptonT79}
Richard~J. Lipton and Robert~Endre Tarjan.
\newblock A separator theorem for planar graphs.
\newblock {\em SIAM J.\ Appl.\ Math.}, 36(2):177--189, 1979.

\bibitem[LT80]{journals/siamcomp/LiptonT80}
Richard~J. Lipton and Robert~Endre Tarjan.
\newblock Applications of a planar separator theorem.
\newblock {\em SIAM J.\ Comput.}, 9(3):615--627, 1980.

\bibitem[Mil86]{journals/jcss/Miller86}
Gary~L. Miller.
\newblock Finding small simple cycle separators for 2-connected planar graphs.
\newblock {\em J.\ Comput.\ Syst. Sci.}, 32(3):265--279, 1986.

\bibitem[MS12]{MS12}
Shay Mozes and Christian Sommer.
\newblock Exact distance oracles for planar graphs.
\newblock In {\em Proceedings of the 23rd Annual ACM-SIAM Symposium on Discrete
  Algorithms (SODA)}, pages 209--222, 2012.

\bibitem[MWN10]{esa/MozesW10}
Shay Mozes and Christian Wulff-Nilsen.
\newblock Shortest paths in planar graphs with real lengths in
  {$O(n\log^2n/\log\log n)$} time.
\newblock In {\em Proceedings of the 18th European Symposium on Algorithms
  (ESA)}, pages 206--217, 2010.

\bibitem[MZ08]{journals/siamcomp/MaheshwariZ08}
Anil Maheshwari and Norbert Zeh.
\newblock {I/O}-efficient planar separators.
\newblock {\em SIAM J.\ Comput.}, 38(3):767--801, 2008.

\bibitem[PRS94]{conf/soda/PlotkinRS94}
Serge~A. Plotkin, Satish Rao, and Warren~D. Smith.
\newblock Shallow excluded minors and improved graph decompositions.
\newblock In {\em Proceedings of the 5th Annual ACM-SIAM Symposium on Discrete
  Algorithms (SODA)}, pages 462--470, 1994.

\bibitem[RH00]{RosenbergHeath}
Arnold~L. Rosenberg and Lenwood~S. Heath.
\newblock {\em Graph Separators, with Applications}.
\newblock Frontiers in Computer Science. Springer, 2000.

\bibitem[RS86]{RS2}
Neil Robertson and Paul~D. Seymour.
\newblock Graph minors. {II}. {Algorithmic} aspects of tree-width.
\newblock {\em J. Alg.}, 7:309--322, 1986.

\bibitem[RW09]{journals/talg/ReedW09}
Bruce~A. Reed and David~R. Wood.
\newblock A linear-time algorithm to find a separator in a graph excluding a
  minor.
\newblock {\em TALG}, 5(4), 2009.

\bibitem[Shm97]{Shmoys}
David~B. Shmoys.
\newblock Cut problems and their application to divide-and-conquer.
\newblock In {\em Approximation algorithms for {NP}-hard problems}. PWS
  Publishing Co., 1997.

\bibitem[ST83]{ST83}
Daniel~Dominic Sleator and Robert~Endre Tarjan.
\newblock A data structure for dynamic trees.
\newblock {\em J.\ Comput.\ Syst. Sci.}, 26(3):362--391, 1983.

\bibitem[ST85]{journals/jacm/SleatorT85}
Daniel~Dominic Sleator and Robert~Endre Tarjan.
\newblock Self-adjusting binary search trees.
\newblock {\em JACM}, 32(3):652--686, 1985.

\bibitem[ST96]{conf/compgeom/SpielmanT96}
Daniel~A. Spielman and Shang-Hua Teng.
\newblock Disk packings and planar separators.
\newblock In {\em Proceedings of the 12th Annual Symposium on Computational
  Geometry (SOCG)}, pages 349--358, 1996.

\bibitem[Sub95]{SubramanianThesis}
Sairam Subramanian.
\newblock {\em Parallel and Dynamic Shortest-Path Algorithms for Sparse
  Graphs}.
\newblock PhD thesis, Brown University, 1995.
\newblock Available as Brown University Computer Science Technical Report
  CS-95-04.

\bibitem[TW05]{TW05}
Robert~Endre Tarjan and Renato~Fonseca Werneck.
\newblock Self-adjusting top trees.
\newblock In {\em Proceedings of the 16th Annual ACM-SIAM Symposium on Discrete
  Algorithms (SODA)}, pages 813--822, 2005.

\bibitem[Ung51]{Ungar51}
Peter Ungar.
\newblock A theorem on planar graphs.
\newblock {\em J.\ Lond.\ Math.\ Soc.}, 1-26(4):256--262, 1951.

\bibitem[vS47]{vonStaudt}
Karl Georg~Christian von Staudt.
\newblock {\em Geometrie der Lage}.
\newblock Bauer und Raspe, N\"urnberg, 1847.

\bibitem[Whi32]{Whitney}
Hassler Whitney.
\newblock Non-separable and planar graphs.
\newblock {\em AMS Trans.}, 34(2):339--362, 1932.

\bibitem[WN11]{CWNMinorFree}
Christian Wulff-Nilsen.
\newblock Separator theorems for minor-free and shallow minor-free graphs with
  applications.
\newblock In {\em Proceedings of the 52nd Annual Symposium on Foundations of
  Computer Science (FOCS)}, pages 37--46, 2011.

\end{thebibliography}

\newpage
\begin{appendix}
\section{Proof of Lemma~\ref{lem:recurrence}}
\begin{proof}
We show by induction on $n$ that, for $n\geq r$, $T_r(n) \leq \delta\gamma n / r^{1-\beta} -
\gamma n^\beta$, where $\gamma>1$ and $\delta>1$ are constants to be
determined. Note that $\delta\gamma n / r^{1-\beta} -
\gamma n^\beta$ is non-negative whenever $\delta n / r^{1-\beta} \geq
n^\beta$, or equivalently, whenever $n \geq \frac{r}{\delta^{1/(1-\beta)}}$.

Suppose $n>r$, and assume the lemma holds for $n'<n$. Fix values of
$\alpha_1, \ldots, \alpha_8$. We show
that $\rho n^\beta + \sum_i T_r(\alpha_i n) \leq \delta\gamma n / r^{1-\beta} -
\gamma n^\beta$. Let $\mathcal I = \set{i \ 
 :\ \alpha_i n > \frac{r}{\delta^{1/(1-\beta)}}}$, and let
$q=\abs{\mathcal I}$.

Case 0: $q=0$. In this case, $\alpha_i n < \frac{r}{\delta^{1/(1-\beta)}}$
for all $1 \leq i \leq 8$. Since $\sum \alpha_i \geq 1$, it must be
that $r < n \leq 8r$.
Hence $T_r(n) \leq \rho n^\beta$, and 
\begin{equation} \label{eq:case1a}
\delta\gamma n / r^{1-\beta} - \gamma n^\beta \geq \delta \gamma r^\beta -
8 \gamma r^\beta 
 = (\delta-8) \gamma r^\beta 
\end{equation}
Setting $\gamma > 8\rho$ we get
\begin{eqnarray} \label{eq:case1b}
T_r(n) & \leq & \rho n^\beta \nonumber \\
& \leq & 8 \rho r^\beta \nonumber \\
& \leq & \gamma r^\beta
\end{eqnarray}
Requiring $\delta \geq 9$ and combining~\eqref{eq:case1a} and
~\eqref{eq:case1b} proves the inductive step in this case.

Case 1: $q \geq 1$. 
For
every $i \notin {\mathcal I}$, 
$$\alpha_i n \leq \frac{r}{\delta^{1/(1-\beta)} }\leq
 \frac{n}{\delta^{1/(1-\beta)}}.$$
Hence, for
every $i \notin {\mathcal I}$, $\alpha_i \leq 
 \frac{1}{\delta^{1/(1-\beta)}}$.
It follows that 
\begin{equation}\label{eq:sum-alpha}
\sum_{i \in {\mathcal I}} \alpha_i \geq 1 - \frac{7}{\delta^{1/(1-\beta)}}
\end{equation}

We claim that $\delta$ can be chosen sufficiently large so that
$\sum_{i \in {\mathcal I}} \alpha_i^\beta = 1 +\epsilon$ for some
 $\epsilon > 0$. To see this, assume without loss of generality that $\alpha_1 \geq {
 \alpha}_2 \geq \dots \geq \alpha_8$, and
observe that
\begin{eqnarray}\label{eq:sumI}
\sum_{i \in {\mathcal I}} \alpha_i^\beta & \geq & {
 \alpha}_1^\beta - \alpha_1 + \sum_{i \in {\mathcal I}} {
 \alpha}_i \nonumber \\
& \geq & {
 \alpha}_1^\beta - \alpha_1 +1 - \frac{7}{\delta^{1/(1-\beta)}}
\end{eqnarray}
Let $s$ be sufficiently large so that $3/4 + c/\sqrt n < 1$. Since $1/8 \leq \alpha_1 \leq 3/4 + c/\sqrt n$, $ {
 \alpha}_1^\beta - \alpha_1$ is strictly positive. Define
 $\epsilon'>0$ to be $ \alpha_1^\beta - {
 \alpha}_1$. Note that, by concavity of $\alpha^\beta - \alpha$,
 $\epsilon'$ is monotonically increasing with $s$. Substituting into~\eqref{eq:sumI} we get
\begin{equation}
\sum_{i \in {\mathcal I}} \alpha_i^\beta \geq 1+ \epsilon' - \frac{7}{\delta^{1/(1-\beta)}}
\end{equation}

Define $\epsilon$ to be $\epsilon' -
\frac{7}{\delta^{1/(1-\beta)}} $, and set $\delta$ sufficiently large
so that $\epsilon > 0$. We have
\begin{equation}
\sum_{i \in {\mathcal I}} \alpha_i^\beta \geq 1+ \epsilon
\end{equation}

For each $i\not \in \mathcal I$, $T_r(\alpha_i n) \leq 0$. Hence
\begin{eqnarray} 
\rho n^\beta + \sum_i T_r(\alpha_i n) & = & 
\rho n^\beta + \sum_{i \in
{\mathcal I}} \left(
 \delta\gamma \alpha_i n/r^{1-\beta} - \gamma \alpha_i^\beta n^\beta
\right) \nonumber \\
& \leq & \rho n^\beta + \delta \gamma n / r^{1-\beta} \sum_{i \in
{\mathcal I}} \alpha_i -
\gamma n^\beta \sum_{i \in
{\mathcal I}} \alpha_i^\beta \nonumber \\
& \leq & \rho n^\beta + (1 + c/\sqrt n) \delta \gamma n / r^{1-\beta} -
\gamma n^\beta(1+\epsilon) \nonumber \\
& = & \delta \gamma n / r^{1-\beta} - \gamma n^\beta + \rho n^\beta
+ \delta \gamma c \sqrt n / r^{1-\beta} -
\epsilon \gamma n^\beta\nonumber \\
& = & \delta \gamma n / r^{1-\beta} - \gamma n^\beta + n^\beta
\left( \rho + \delta \gamma c n^{1/2-\beta} / r^{1-\beta} -
\epsilon \gamma \right)\nonumber \\
& \leq & \delta \gamma n / r^{1-\beta} - \gamma n^\beta + n^\beta
\left( \rho + \delta \gamma c / r^{1-\beta} -
\epsilon \gamma \right) \label{eq:lastTr}
\end{eqnarray}
Let $s$ be sufficiently large so that for any $r\geq s$,
$\delta c/r^{1-\beta} < \epsilon/2$. 
Note that $\epsilon$ does depend on $s$, but it since it is monotonically increasing
with $s$, there exists such sufficiently large $s$. 

Then the right-hand side
of~\eqref{eq:lastTr} is at most

$$\delta \gamma n / r^{1-\beta} - \gamma n^\beta + n^\beta \left( \rho -
\epsilon \gamma /2 \right) $$
Setting $\gamma > \rho/\epsilon$ completes the proof.
\end{proof}

\end{appendix}
\end{document}